\documentclass[12pt]{article}
\usepackage{comment} 
\usepackage[T1]{fontenc}
\usepackage[utf8]{inputenc}
\usepackage[a4paper,margin=1in]{geometry}
\usepackage{amsmath,amssymb,amsthm,mathtools}
\usepackage{bm}
\usepackage{graphicx}
\usepackage{tikz}
\usetikzlibrary{calc}
\usepackage{booktabs}
\usepackage{tabularx}
\usepackage{float}
\usepackage{subcaption}
\usepackage{xcolor}
\usepackage{listings}
\usepackage[numbers,sort&compress]{natbib}
\usepackage{hyperref}
\usepackage{hypernat}
\usepackage{fancyhdr}
\hypersetup{colorlinks=true,
  linkcolor=blue!55!black, citecolor=blue!55!black, urlcolor=blue!55!black}
 
\theoremstyle{plain}
\newtheorem{theorem}{Theorem}
\newtheorem{proposition}[theorem]{Proposition}
\newtheorem{lemma}[theorem]{Lemma}

\theoremstyle{definition}
\newtheorem{definition}[theorem]{Definition}
\newtheorem{remark}[theorem]{Remark}

\newcommand{\bF}{\bm{F}}
\newcommand{\bC}{\bm{C}}
\newcommand{\bE}{\bm{E}}

\newcommand{\bs}{\bm{s}}
\newcommand{\bmm}{\bm{m}}
\newcommand{\bu}{\bm{u}}
\newcommand{\by}{\bm{y}}
\newcommand{\bx}{\bm{x}}
\newcommand{\bI}{\bm{I}}
\newcommand{\bN}{\bm{N}}
\newcommand{\bK}{\bm{K}}
\newcommand{\bbeta}{\bm{\beta}}
\newcommand{\Fe}{\bF^{e}}
\newcommand{\Fp}{\bF^{p}}
\newcommand{\Ce}{\bC^{e}}
\newcommand{\Pe}{\bm{P}^{e}}
\newcommand{\bM}{\bm{M}}
\newcommand{\R}{\mathbb{R}}
\newcommand{\tr}{\operatorname{tr}}
\newcommand{\cof}{\operatorname{cof}}
\DeclareMathOperator*{\argmin}{arg\,min}
 
\definecolor{cobalt}{rgb}{0.0,0.28,0.67}
\begin{document}
\title{Formation of grain boundaries in ductile single crystals under
plane-strain simple shear: a block-coordinate finite element method}
\author{Khanh Chau Le$^{a,b}$\thanks{Corresponding author. Email: lekhanhchau@tdtu.edu.vn}
  \ and \ Thanh Danh Nguyen$^{c}$\\[6pt]
  \small $^{a}$Mechanics of Advanced Materials and Structures, Institute for Advanced   Study in Technology,\\
  \small Ton Duc Thang University, Ho Chi Minh City, Vietnam\\
  \small $^{b}$Faculty of Civil Engineering, Ton Duc Thang University, Ho Chi Minh City, Vietnam\\
  \small $^{c}$Faculty of Mathematics and Statistics, Ton Duc Thang University, Ho Chi Minh City, Vietnam}
\date{}
 
\maketitle

\begin{abstract}
\noindent
Large plastic deformation can drive an initially uniform single crystal to
spontaneously subdivide into misoriented grains separated by thin dislocation
walls -- a pattern-forming instability rooted in the loss of convexity of the
crystal's elastic energy at large strain. We study this phenomenon for a
ductile crystal in plane-strain simple shear within continuum dislocation
theory, using a polyconvex (Ciarlet--Geymonat) elastic energy that guarantees
existence of minimizers for the coupled deformation--slip problem. Minimizing
over the plastic slip yields a condensed energy of double-well form whose
non-quasiconvexity favours a lamellar microstructure; the gradient of the
geometrically necessary dislocation density regularizes it, giving the grain
boundaries a finite thickness and energy as functions of the misorientation
angle. A block-coordinate finite element scheme -- alternating a convex
non-smooth solve for the slip with a Levenberg-regularized Newton solve for
the deformation -- resolves this microstructure numerically and detects its
spontaneous onset, reproducing the lamellar grain structure in agreement with
the closed-form analysis.
 
\end{abstract}

\section{Introduction}\label{sec:intro}

Many solids, when driven far enough from their reference state, do not deform uniformly but instead break spontaneously into fine-scale patterns --- twinned martensite, wrinkled sheets, laminated shear bands --- because the elastic energy governing the deformation loses convexity along the loading path. The mathematics of such transitions is by now classical: once the stored-energy density fails to be quasiconvex, minimizing sequences develop ever finer oscillations rather than converging to a smooth deformation, and the observed microstructure is best understood as a minimizer of the relaxed (quasiconvexified) energy. This paper studies one instance of this broader phenomenon --- the spontaneous formation of grain boundaries in a ductile single crystal under severe plastic shear --- combining a variational-existence result for the underlying elastic energy with a general-purpose numerical scheme for resolving the resulting microstructure.

Under severe plastic deformation --- as produced, for instance, by equal-channel angular pressing \citep{segal1995materials,valiev2006principles} --- initially defect-poor single crystals subdivide into grains whose boundaries, unlike sharp interfaces, have a finite thickness and contain large numbers of dislocations. Deformed metals develop cell blocks delimited by extended planar dislocation walls --- the geometrically necessary boundaries, whose misorientation grows with strain --- alongside more random incidental dislocation boundaries \citep{kuhlmann1991geometrically,hughes1997high}. The hypothesis that these patterns are configurations of least energy among those accessible to the imposed deformation is the low-energy dislocation structure (LEDS) hypothesis of \citet{kuhlmann1989theory}; it motivates treating grain formation as a problem of energy minimization. Within continuum dislocation theory the geometrically necessary boundaries have been modeled along exactly these lines \citep{koster2015bformation,koster2015aformation}, emerging as the low-energy structure separating regions of nearly uniform plastic slip combined with lattice rotation.

The driving mechanism is the loss of convexity of the crystal energy at large plastic strains: the condensed energy obtained after eliminating the plastic slip is a multi-well function whose minimization admits no classical minimizer. \citet{ortiz1999nonconvex} showed that the incremental (pseudoelastic) energy of a crystal undergoing geometrical softening or latent hardening fails to be quasiconvex, with wells corresponding to single-slip deformations; they constructed laminate microstructures that undercut the homogeneous state and identified their interfaces with dislocation walls. This was developed into a quantitative theory of subgrain dislocation structures by \citet{ortiz2000theory} and \citet{aubry2003the}. \citet{carstensen2002non} proved that the loss of quasiconvexity persists even without geometrical softening or latent hardening once slip is confined to a single system — the setting adopted here. The mathematical analysis of the single-slip problem is now well developed \citep{dacorogna2008direct,conti2005single,conti2005dislocation,conti2016relaxation,conti2021optimal}.

A finite microstructural scale requires regularization, supplied here by continuum dislocation theory \citep{nye1953some,bilby1955types,kroner1955fundamentale,berdichevsky1967dynamic,berdichevsky2006continuum,bilby1958continuous,berdichevsky2006thermodynamics}. \citet{koster2015bformation} proposed a model in which the energy of the dislocation network --- depending on the gradient of the plastic slip --- regularizes the non-convex minimization and selects boundaries of finite thickness, with thickness and energy computed as functions of the misorientation angle; \citet{koster2015aformation} developed the fully nonlinear theory, interpreting grain boundaries as surfaces of weak discontinuity in placement but strong discontinuity in plastic slip. A complementary line of work, the thermodynamic dislocation theory \citep{le2018dthermodynamic,piao2021dislocation,piao2022thermodynamic}, describes the interaction of \emph{existing} grain boundaries with dislocations and the attendant work hardening; the present paper concerns instead the \emph{formation} of the boundaries.

Two limitations of the published treatments motivate the present work. First, \citet{koster2015aformation,koster2015bformation} adopt a Saint-Venant--Kirchhoff (SVK) shape energy that is not polyconvex, so the deformation subproblem carries no existence theory; polyconvexity \citep{ball1976convexity} resolves this. The requirement was brought into the present setting in the dissertation of \citet{koster2018modeling}, who demonstrated the formation of the lamellar microstructure with a polyconvex compressible neo-Hookean energy and noted the Ciarlet–Geymonat construction \citep{ciarlet1982lois,ciarlet1988mathematical} as the polyconvex family reproducing SVK behavior at small strains; our treatment builds on that step, adopting the Ciarlet--Geymonat energy itself, whose consequences --- the degree-four condensed energy and the closed-form boundary-layer potential below --- are particular to this choice. To our knowledge, the Ciarlet--Geymonat energy has not previously been employed in continuum dislocation theory. Second, the analytical route is tied to a one-dimensional ansatz at a fixed slip orientation and does not extend to general geometries, loadings, or automatic detection of the onset of microstructure; removing this restriction calls for a numerical minimization of the full functional, in the spirit of existing computational treatments of non-convex crystal-plasticity energies \citep{miehe2004analysis,kochmann2011the,vidyasagar2018deformation,kumar2020assessment,arora2020dislocation,moulinec1998numerical}.

This paper addresses both limitations. Its contributions are as follows. (i) \emph{The variational analysis of the Ciarlet–Geymonat crystal}: polyconvexity and coercivity of the elastic energy, existence for the coupled minimization over deformation and plastic slip (Proposition~\ref{prop:existence}), and linearization to the identical isotropic Hooke law; the condensed energy is a degree-four double well and the boundary-layer potential is obtained in closed form. (ii) \emph{A general block-coordinate finite element scheme}, combining a convex alternating direction method of multipliers  (ADMM) solve for the plastic slip with a Levenberg-regularized Newton solve for the deformation, not restricted to a one-dimensional ansatz or a special slip orientation. (iii) \emph{An energy-based detection of the onset of microstructure}, via a random perturbation of the homogeneous state accepted only when it lowers the energy.

The paper is organized as follows. Section~\ref{sec:CG} establishes the polyconvexity and small-strain limit of the Ciarlet--Geymonat energy; Section~\ref{sec:var} sets up the variational problem;
Section~\ref{sec:shear} specializes to plane-strain simple shear;
Section~\ref{sec:reduced} minimizes over the plastic slip and derives the
condensed energy; Section~\ref{sec:relax} analyzes its non-quasiconvexity and the
lamellar relaxation; Section~\ref{sec:GNB} regularizes the boundaries to a finite
thickness; Section~\ref{sec:method} develops the block-coordinate finite element
method and reports the numerical results; and Section~\ref{sec:conclusions}
concludes.

\section{The Ciarlet--Geymonat energy density}\label{sec:CG}

\subsection{Setting and notation}\label{sec:setting}
Let $\bF=\partial\by/\partial\bx$ denote the total deformation gradient, decomposed
multiplicatively into elastic and plastic parts,
\begin{equation}\label{eq:FeFp}
  \bF=\Fe\,\Fp ,\qquad \Fe=\bF\,(\Fp)^{-1},
\end{equation}
where $\Fp$ describes the plastic slip carried by dislocations and leaves the
lattice undistorted, while $\Fe$ distorts the lattice and carries the elastic
energy. Plastic slip is isochoric, $\det\Fp=1$, so that
\begin{equation}\label{eq:J}
  J:=\det\Fe=\det\bF .
\end{equation}
The elastic right Cauchy-Green tensor and its first invariant are
\begin{equation}\label{eq:invariants}
  \Ce=(\Fe)^{\!\top}\Fe ,\qquad I_1=\tr\Ce .
\end{equation}

\subsection{The energy and its polyconvexity}\label{sec:polyconvex}
The elastic part of the stored energy density is the Ciarlet--Geymonat form
\begin{equation}\label{eq:psi-el}
  \psi_{\mathrm{el}}(\Fe)=\frac{\mu}{2}\bigl(I_1-3-2\ln J\bigr)
  +\frac{\lambda}{2}(J-1)^2,
\end{equation}
with $I_1=\tr\Ce$, $J=\det\Fe>0$, and Lam\'e constants $\mu>0$, $\lambda\ge0$.
We record two classical properties
of \eqref{eq:psi-el}, quoted rather than re-derived.

\begin{definition}[Polyconvexity, after \citet{ball1976convexity}]\label{def:polyconvex}
A stored-energy density $W(\bF)$, defined for $\det\bF>0$, is
\emph{polyconvex} if there exists a function $g(\bF,\bm{M},\delta)$, jointly
convex on $\R^{3\times3}\times\R^{3\times3}\times(0,\infty)$, such that
\begin{equation}\label{eq:polyconvex-def}
  W(\bF)=g\bigl(\bF,\ \cof\bF,\ \det\bF\bigr).
\end{equation}
\end{definition}

\begin{proposition}[Polyconvexity of the elastic energy]\label{prop:polyconvex}
Let $\mu>0$ and $\lambda\ge0$. The elastic energy \eqref{eq:psi-el}
is polyconvex in the sense of Definition~\ref{def:polyconvex}. Moreover
$\psi_{\mathrm{el}}(\Fe)\to+\infty$ as $J\to0^{+}$.
\end{proposition}

\noindent The proof is the classical additive convex split
$\psi_{\mathrm{el}}(\bF)=\tfrac{\mu}{2}\|\bF\|^{2}-\tfrac{3\mu}{2}
+\Gamma(\det\bF)$ with
$\Gamma(\delta)=-\mu\ln\delta+\tfrac{\lambda}{2}(\delta-1)^{2}$ convex on
$(0,\infty)$, due to \citet{ciarlet1982lois} and developed in
\citet{ciarlet1988mathematical}; it is not reproduced here. 

\begin{lemma}[Lower bound and coercivity, $\lambda\ge0$]\label{lem:coercive}
Let $\mu>0$ and $\lambda\ge0$. For every $c_1\in(0,\mu/2)$ there is a constant
$c_0=c_0(\mu,c_1)\ge0$ such that
\begin{equation}\label{eq:coercive}
  \psi_{\mathrm{el}}(\Fe)\ \ge\ c_1\|\Fe\|^{2}-c_0
  \qquad\text{for every }\Fe\text{ with }\det\Fe>0 .
\end{equation}
In particular $\psi_{\mathrm{el}}$ is bounded below and coercive in $\|\Fe\|$.
\end{lemma}

\noindent The estimate is standard \citep{ciarlet1988mathematical}: drop the nonnegative
volumetric term and bound $\ln J\le\tfrac32\ln(\|\Fe\|^{2}/3)$ by the Hadamard
and arithmetic-geometric mean inequalities - a logarithm grows more slowly than
$\|\Fe\|^{2}$. 

\begin{remark}\label{rem:notconvex}
In plane strain the in-plane gradient $\bm{A}$ satisfies
$\|\cof\bm{A}\|=\|\bm{A}\|$, so Lemma~\ref{lem:coercive} gives full coercivity;
with Proposition~\ref{prop:polyconvex} the direct method \citep{ball1976convexity}
yields a $\by$-minimizer. As $\psi_{\mathrm{el}}$ is frame-indifferent
($I_1,J$ are rotation-invariant), it is polyconvex but \emph{not} convex - so
this minimizer need not be unique.
\end{remark}

\subsection{Linearization at the natural configuration}\label{sec:linear}

\begin{proposition}[Linearization to Hooke's law]\label{prop:linear}
Write $\Fe=\bI+\bm{H}$ with $\bm{H}=\nabla\bu$, and let
$\bm{\varepsilon}_e=\tfrac12(\bm{H}+\bm{H}^{\!\top})$ be its symmetric part.
Then, as $\|\bm{H}\|\to0$,
\begin{equation}\label{eq:hooke-limit}
  \psi_{\mathrm{el}}(\bI+\bm{H})=\mu\,\bm{\varepsilon}_e\!:\!\bm{\varepsilon}_e
  +\frac{\lambda}{2}\,(\tr\bm{\varepsilon}_e)^{2}+O(\|\bm{H}\|^{3}),
\end{equation}
i.e.\ the model reduces to isotropic linear elasticity with Lam\'e constants
$\mu,\lambda$.
\end{proposition}

\noindent This is the standard Taylor expansion at $\bF=\bI$ - the
Ciarlet--Geymonat family is constructed precisely so that its linearization
reproduces a prescribed isotropic Hooke law
\citep{ciarlet1982lois,ciarlet1988mathematical} - and is not reproduced here.

\begin{remark}\label{rem:hooke}
The limit \eqref{eq:hooke-limit} is the isotropic Hooke energy; its stress
$2\mu\,\bm{\varepsilon}_e+\lambda(\tr\bm{\varepsilon}_e)\bI$ identifies
$\mu,\lambda$ as the Lam\'e constants, so no parameter re-fitting is needed
between the finite- and small-strain regimes. Together with
Proposition~\ref{prop:polyconvex}, this justifies the use of \eqref{eq:psi-el}:
it is well-posed in the finite-strain regime and physically correct in the
small-strain limit.
\end{remark}

\section{Variational formulation of grain boundary formation}\label{sec:var}

\subsection{Kinematics of single slip}\label{sec:kinematics}
We use the single-slip continuum-dislocation kinematics of
\citet{koster2015aformation} in the plane-strain geometry of
Figure~\ref{fig:setup}. For a single active slip system with unit slip direction
$\bs$ and unit slip-plane normal $\bmm$ ($\bs\cdot\bmm=0$), the plastic
distortion is
\begin{equation}\label{eq:Fp-slip}
  \Fp=\bI+\beta\,\bs\otimes\bmm ,
  \qquad \det\Fp=1 ,
\end{equation}
where $\beta$ is the plastic slip.
The incompatibility of $\Fp$ is measured by Nye's dislocation density tensor
\citep{nye1953some,bilby1955types,kroner1955fundamentale}
\begin{equation}\label{eq:nye}
  \bm{T}=-\Fp\times\nabla=\bs\otimes(\nabla\beta\times\bmm).
\end{equation}
For the plane-strain single slip system
all dislocation lines are parallel to the out-of-plane axis $\bm{l}$, so the
scalar geometrically necessary dislocation density $\rho=|\bm{T}\cdot\bm{l}|/b$
reduces to the form linear in $\nabla\beta$
\begin{equation}\label{eq:rho}
  \rho=\frac1b\,\bigl|\nabla\beta\cdot\bs\bigr| ,
\end{equation}
with $b$ the magnitude of the Burgers vector, since
$\nabla\beta\times\bmm=(\nabla\beta\cdot\bs)\,\bm{l}$. Physically, geometrically necessary dislocations are
those the lattice must store to accommodate a \emph{gradient} of plastic slip:
a uniform slip carries none, and $\rho$ grows only where $\beta$ varies in space -
the reason the dislocation energy below sets a finite grain-boundary width.

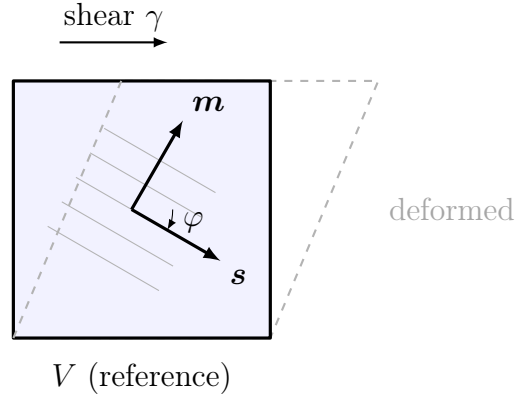
\begin{figure}[h]\centering
\begin{tikzpicture}[scale=3.4,>=latex,line join=round]
  \def\g{0.42}
  \draw[dashed,gray!60,thick] (0,0)--(1,0)--(1+\g,1)--(\g,1)--cycle;
  \fill[blue!5] (0,0) rectangle (1,1);
  \draw[very thick] (0,0) rectangle (1,1);
  \draw[dashed,gray!60,thick] (0,0)--(\g,1);
  \node[below] at (0.5,-0.05) {$V$ (reference)};
  \node[gray!70,right] at (1+\g,0.5) {deformed};
  \draw[->,thick] (0.18,1.15)--(0.18+\g,1.15) node[midway,above] {shear $\gamma$};
  \coordinate (O) at (0.46,0.5);
  \def\ph{-30}
  \foreach \d in {-0.22,-0.11,0,0.11,0.22}{
    \draw[thin,gray!60]
      ($(O)+(\ph+90:\d)+(\ph:-0.25)$)--($(O)+(\ph+90:\d)+(\ph:0.25)$);
  }
  \draw[->,very thick] (O)--($(O)+(\ph:0.40)$)    node[below right=-1pt] {$\bs$};
  \draw[->,very thick] (O)--($(O)+(\ph+90:0.40)$) node[above right=-1pt] {$\bmm$};
  \draw[->] ($(O)+(0.16,0)$) arc (0:\ph:0.16);
  \node at ($(O)+(0.24,-0.05)$) {$\varphi$};
\end{tikzpicture}
\caption{Plane-strain simple shear of the unit cell $V$: the reference square
(solid) is mapped to the parallelogram (dashed) by
$\bF=\bI+\gamma\,\bm{e}_1\otimes\bm{e}_2$. The single slip system has direction
$\bs$ and plane normal $\bmm$ at orientation $\varphi$; the thin gray lines are
the slip planes. Hard-device data $\by=\bF\bx$, $\beta=0$ are imposed on
$\partial V$.}
\label{fig:setup}
\end{figure}

\subsection{Free energy}\label{sec:energy}
The stored energy density per unit reference volume is
\begin{equation}\label{eq:psi}
  \psi \;=\; \underbrace{\frac{\mu}{2}\bigl(I_1-3-2\ln J\bigr)
  \;+\; \frac{\lambda}{2}\,(J-1)^2}_{\text{Ciarlet--Geymonat elastic}}
  \;+\; \underbrace{\mu k\,\frac{\rho}{\rho_s}
  \;+\; \frac{\mu k}{2}\Bigl(\frac{\rho}{\rho_s}\Bigr)^{2}}_{\text{dislocation network}} ,
\end{equation}
with $\mu>0$, $\lambda\ge0$, dimensionless coefficient $k>0$, saturated density
$\rho_s$, and $\rho=\tfrac1b|\nabla\beta\cdot\bs|$ from \eqref{eq:rho}. The
quadratic term coincides with the moderate-density dislocation-network energy of
\citet[Eq.~(1)]{koster2015bformation}; the linear term has the form of the dislocation line energy $\mu c b^{2}\rho$ of \citet[Eq.~(25)]{koster2015aformation} and sharpens the dislocation walls. For very small densities, and for densities close to
saturation, a saturation-type (logarithmic) energy is more appropriate
\citep{berdichevsky2006thermodynamics}; here we retain the
linear and quadratic terms. Introducing the internal length
$\ell=1/(b\rho_s)$, so that $\rho/\rho_s=\ell\,|\nabla\beta\cdot\bs|$, the network
density reads
\begin{equation}\label{eq:disl-moduli}
  q_c\,\bigl|\nabla\beta\cdot\bs\bigr|
  +\tfrac12 c_2\,(\nabla\beta\cdot\bs)^2 ,
  \qquad q_c=\mu k\,\ell,\quad c_2=\mu k\,\ell^{2}.
\end{equation}
The \emph{linear} term $q_c|\nabla\beta\cdot\bs|$ is of bounded-variation (BV) type --- its integral is the total variation of $\beta $ along the slip direction --- and, being non-smooth in $\nabla \beta $ it sharpens the dislocation walls; the \emph{quadratic} term $c_2$ regularizes them to a finite thickness, so that the admissible slips remain in the space $H_{\bs}(V)$ of Section 3.4 and no actual jump discontinuities occur. The present model thus combines the two dislocation contributions with a \emph{polyconvex} elastic energy: the Ciarlet--Geymonat energy replaces the Saint-Venant--Kirchhoff shape energy
$\tfrac12\lambda(\tr\bE^{e})^{2}+\mu\,\tr\bigl((\bE^{e})^{2}\bigr)$,
$\bE^{e}=\tfrac12(\Ce-\bI)$, of the earlier treatments, while the network energy
retains their quadratic term together with a line-tension
(bounded-variation) term of the same type.

\subsection{Energy functional and boundary conditions}\label{sec:variational}
We minimize the energy functional
\begin{equation}\label{eq:functional}
  I[\by(\bx),\beta(\bx)]=\int_{V} W\bigl(\bF,\beta,\nabla\beta\bigr)\,d\bx ,
\end{equation}
where $W(\bF,\beta,\nabla\beta):=\psi(\Fe,\rho)$ is the density \eqref{eq:psi}
expressed through $(\bF,\beta,\nabla\beta)$ by
$\Fe=\bF(\bI-\beta\,\bs\otimes\bmm)$ and $\rho=\tfrac1b|\nabla\beta\cdot\bs|$,
subject to the boundary conditions (hard device, prescribed \emph{total}
deformation $\bF$)
\begin{equation}\label{eq:bc}
  \by(\bx)=\bF\bx ,\qquad \beta(\bx)=0
  \qquad\text{on }\partial V .
\end{equation}
Stationarity in $\by$ yields the equilibrium equations, and stationarity in
$\beta$ the yield/relaxation condition; both are made explicit and discretized in
Section~\ref{sec:method}.

\subsection{Existence of minimizers for the coupled problem}\label{sec:existence}

We now show that the regularized functional \eqref{eq:functional} attains its
minimum jointly in $(\by,\beta)$. Throughout this subsection we work in the
plane-strain setting (cf.\ Remark~\ref{rem:notconvex}): $V\subset\R^{2}$ is a
bounded Lipschitz domain, $\by\colon V\to\R^{2}$, all tensors are the in-plane
$2\times2$ blocks, and $\{\bs,\bmm\}$ is an orthonormal basis of $\R^{2}$.
Since only the directional derivative $\nabla\beta\cdot\bs$ enters the
dislocation energy \eqref{eq:psi}, the natural slip space is anisotropic,
\begin{equation}\label{eq:Hs}
  H_{\bs}(V):=\bigl\{\beta\in L^{2}(V):\ \nabla\beta\cdot\bs\in L^{2}(V)\bigr\},
\end{equation}
a Hilbert space with the norm
$\bigl(\|\beta\|_{L^{2}}^{2}+\|\nabla\beta\cdot\bs\|_{L^{2}}^{2}\bigr)^{1/2}$;
we write $B_{0}$ for the closure of $C_{c}^{\infty}(V)$ in $H_{\bs}(V)$, the
weak form of the hard-device condition $\beta=0$ on $\partial V$ in
\eqref{eq:bc}. This condition is meaningful because every line parallel to the
slip direction exits the domain:

\begin{lemma}[Directional Poincar\'e inequality]\label{lem:poincare}
For every $\beta\in B_{0}$,
\begin{equation}\label{eq:poincare}
  \|\beta\|_{L^{2}(V)}\ \le\ \operatorname{diam}(V)\,
  \|\nabla\beta\cdot\bs\|_{L^{2}(V)} .
\end{equation}
\end{lemma}

\begin{proof}
For $\beta\in C_{c}^{\infty}(V)$, extended by zero outside $V$, integration
along the line through $\bx$ parallel to $\bs$ gives
$\beta(\bx)=\int_{-\infty}^{0}(\nabla\beta\cdot\bs)(\bx+t\bs)\,dt$, whence
$|\beta(\bx)|^{2}\le\operatorname{diam}(V)\int_{\ell(\bx)}|\nabla\beta\cdot\bs|^{2}$
by the Cauchy--Schwarz inequality, with $\ell(\bx)$ the chord of $V$ through
$\bx$ in the direction $\bs$. Integrating over $V$ in coordinates aligned with
$(\bs,\bmm)$ and using Fubini's theorem yields \eqref{eq:poincare}, which
extends to $B_{0}$ by density.
\end{proof}

For $M>0$ the admissible set is
\begin{equation}\label{eq:admissible}
  \mathcal{A}_{M}:=\bigl\{(\by,\beta):\ \by\in H^{1}(V;\R^{2}),\
  \by=\bF\bx\ \text{on }\partial V,\ \beta\in B_{0},\
  |\beta|\le M\ \text{a.e.\ in }V\bigr\},
\end{equation}
with $\bF$ the prescribed boundary deformation, $\det\bF>0$.

\begin{proposition}[Existence for the coupled problem]\label{prop:existence}
Let $\mu>0$, $\lambda>0$ and $M>0$. Then the functional $I$ of
\eqref{eq:functional} attains its minimum on $\mathcal{A}_{M}$, and every
minimizer satisfies $\det\nabla\by>0$ a.e.\ in $V$.
\end{proposition}

\begin{proof}
The homogeneous pair $(\bF\bx,0)\in\mathcal{A}_{M}$ has finite energy, and $I$
is bounded below on $\mathcal{A}_{M}$: the shape term is bounded below by
$-3\mu/2$, the volumetric integrand
$\Gamma(\delta)=-\mu\ln\delta+\tfrac{\lambda}{2}(\delta-1)^{2}$ is convex and
bounded below for $\lambda>0$, and the slip terms are nonnegative. Let
$(\by_{k},\beta_{k})\subset\mathcal{A}_{M}$ be a minimizing sequence.

\emph{Step 1: bounds.} By Lemma~\ref{lem:coercive},
$\|\Fe_{k}\|_{L^{2}}\le C$. Since the plastic distortion inverts explicitly
(Proposition~\ref{prop:convex-beta}), $\nabla\by_{k}=\Fe_{k}(\bI+\beta_{k}\,
\bs\otimes\bmm)$, and $|\beta_{k}|\le M$ gives
$\|\nabla\by_{k}\|_{L^{2}}\le(1+M)\|\Fe_{k}\|_{L^{2}}\le C$; with the boundary
condition, $(\by_{k})$ is bounded in $H^{1}(V;\R^{2})$. The slip energy bounds
$\nabla\beta_{k}\cdot\bs$ in $L^{2}$, and $|\beta_{k}|\le M$ bounds
$\beta_{k}$ in $L^{\infty}$. For $\lambda>0$ the volumetric term bounds
$\det\nabla\by_{k}$ in $L^{2}$.

\emph{Step 2: limits.} Along a subsequence, $\by_{k}\rightharpoonup\by$ in
$H^{1}$ and $\by_{k}\to\by$ in $L^{2}$ (Rellich);
$\beta_{k}\overset{\ast}{\rightharpoonup}\beta$ in $L^{\infty}$ and
$\nabla\beta_{k}\cdot\bs\rightharpoonup\nabla\beta\cdot\bs$ in $L^{2}$;
$\det\nabla\by_{k}\rightharpoonup\det\nabla\by$ in $L^{2}$, by the weak
continuity of the two-dimensional Jacobian --- the distributional identity
$\det\nabla\by=\partial_{1}(y^{1}\partial_{2}y^{2})
-\partial_{2}(y^{1}\partial_{1}y^{2})$ pairs the strongly convergent $\by_{k}$
with the weakly convergent $\nabla\by_{k}$ --- together with the $L^{2}$ bound
of Step~1. The limit is admissible: the trace condition and the set
$B_{0}\cap\{|\beta|\le M\}$ are convex and closed, hence weakly closed.

\emph{Step 3: identification of the elastic distortion.} The only nonlinear
coupling is the product $\beta_{k}\,\partial_{\bs}\by_{k}$ in
$\Fe_{k}=\nabla\by_{k}-\beta_{k}(\nabla\by_{k}\,\bs)\otimes\bmm$.
Distributionally,
\begin{equation*}
  \beta_{k}\,\partial_{\bs}\by_{k}
  =\partial_{\bs}(\beta_{k}\by_{k})-\by_{k}\,(\nabla\beta_{k}\cdot\bs),
\end{equation*}
and both products on the right pair a strongly convergent factor ($\by_{k}$)
with a weakly(-$\ast$) convergent one ($\beta_{k}$, respectively
$\nabla\beta_{k}\cdot\bs$), so
$\beta_{k}\,\partial_{\bs}\by_{k}\to\beta\,\partial_{\bs}\by$ in
$\mathcal{D}'(V)$. Hence the $L^{2}$-bounded sequence $\Fe_{k}$ converges
weakly in $L^{2}$ to $\Fe=\nabla\by\,(\bI-\beta\,\bs\otimes\bmm)$. The
compensated structure is supplied by the model itself:
$\nabla\beta\cdot\bs=\nabla\cdot(\beta\bs)$ is precisely the quantity the
dislocation energy controls.

\emph{Step 4: lower semicontinuity.} By the convex split of
Proposition~\ref{prop:polyconvex},
$\psi_{\mathrm{el}}=\tfrac{\mu}{2}\|\Fe\|^{2}-\tfrac{3\mu}{2}
+\Gamma(\det\nabla\by)$ with $\Gamma$ convex, and the slip energy is convex in
$\nabla\beta\cdot\bs$. Every term of $I$ is therefore a convex integrand of a
weakly converging argument --- $\Fe$ in $L^{2}$, $\det\nabla\by$ in $L^{2}$,
$\nabla\beta\cdot\bs$ in $L^{2}$ --- hence weakly lower semicontinuous by the
standard semicontinuity theorem for convex integrands
\citep{dacorogna2008direct}. Thus $I[\by,\beta]\le\liminf_{k}
I[\by_{k},\beta_{k}]=\inf_{\mathcal{A}_{M}}I$, so $(\by,\beta)$ is a
minimizer, and finiteness of $\int_{V}\Gamma(\det\nabla\by)\,d\bx$ forces
$\det\nabla\by>0$ a.e.
\end{proof}

\begin{remark}[The slip bound and the case $\lambda=0$]\label{rem:slipbound}
(i) The bound $|\beta|\le M$ is physically innocuous: slips beyond the second
well carry no meaning within the dislocation-energy model, and the bound is
inactive at the computed minimizers --- along the simple-shear path
$|\beta^{\ast}|=\gamma^{2}/\bigl((\gamma-1)^{2}+1\bigr)\le2$, and in general it
suffices to take $M>|\beta_{B}|=2|\cot\varphi|$ with a modest margin.
Mathematically the bound substitutes for the transverse compactness that the
purely directional gradient control cannot provide. (ii) The theorem covers
the physically relevant regime $\lambda>0$ (arbitrarily small $\lambda$
suffices); the computations of Section~\ref{sec:results} set $\lambda=0$ for
simplicity, in which case the volumetric control weakens to $-\mu\ln J$ and
the passage to the limit in the determinant term requires additional care.
(iii) Uniqueness cannot be expected (Remark~\ref{rem:notconvex}): consistently
with the non-quasiconvexity of the condensed energy
(Proposition~\ref{prop:nonqc}), the minimizer is a laminate of finite period
selected by the gradient term (Remark~\ref{rem:micro} and
Section~\ref{sec:results}).
\end{remark}

\begin{remark}[Relation to gradient plasticity]\label{rem:existence-lit}
Existence in finite-strain gradient plasticity is usually obtained by
regularizing the full gradient of the plastic distortion
\citep{mainik2009global}; \cite{anguige2018existence} proved existence for the single-slip-to-single-plane relaxation with $L^{p}$-hardening penalty, likewise by div--curl techniques, for the relaxation of the hard single-slip condition itself see \citep{anguige2014relaxation}. The present setting regularizes only the slip derivative along the slip direction --- the geometrically necessary dislocation density available to the single slip system, Eq.~\eqref{eq:rho} --- and the proof shows that this minimal, physically motivated regularization already suffices in plane strain.
\end{remark}

\section{Plane-strain simple shear}\label{sec:shear}
We consider plane-strain simple shear with overall shear $\gamma$, for which the
total deformation gradient is
\begin{equation}\label{eq:shear-F}
  \bF=\bI+\gamma\,\bm{e}_1\otimes\bm{e}_2 ,
\end{equation}
and a single slip system whose slip direction makes an angle $\varphi$ with the
$x_1$-axis,
\begin{equation*}
  \bs=(\cos\varphi,\ \sin\varphi,\ 0),\qquad
  \bmm=(-\sin\varphi,\ \cos\varphi,\ 0).
\end{equation*}
Because $\det\Fp=1$ and the shear is isochoric, $J=\det\Fe=\det\bF=1$ throughout,
so the volumetric terms $-2\ln J$ and $(J-1)^2$ of the Ciarlet--Geymonat
energy \eqref{eq:psi-el} vanish and the homogeneous elastic energy reduces to
$\tfrac{\mu}{2}(I_1-3)$. Consequently the driving force depends only on the shear
modulus $\mu$. The representative orientation $\varphi=-\pi/4$ places the two
energy wells at $\gamma=0$ and $\gamma=2$ (Section~\ref{sec:reduced}); a general
$\varphi$ shifts the second well to $\gamma_B=-2\cot\varphi$ and endows the
resulting grain boundary with the misorientation
$\theta=2\varphi+\pi$ (Section~\ref{sec:GNB}).

\section{Reduced energy and minimization over the plastic slip}\label{sec:reduced}

\subsection{Convexity in the plastic slip}\label{sec:convexbeta}

\begin{proposition}[Convexity in $\beta$ and uniqueness]\label{prop:convex-beta}
Fix the deformation $\by$, hence $\bF$. Then
\begin{enumerate}\setlength{\itemsep}{4pt}
  \item[\textnormal{(i)}] $\Fe$ is affine in $\beta$, namely
  $\Fe=\bF(\bI-\beta\,\bs\otimes\bmm)=\bF-\beta\,(\bF\bs)\otimes\bmm$;
  \item[\textnormal{(ii)}] $J=\det\Fe=\det\bF$ is independent of $\beta$;
  \item[\textnormal{(iii)}] $W$ is convex in $(\beta,\nabla\beta)$: a strictly
  convex quadratic plus the convex \emph{non-smooth} term
  $q_c|\nabla\beta\cdot\bs|$.
\end{enumerate}
Consequently $\beta\mapsto I[\by,\beta]$ is strictly convex and has a
\emph{unique} minimizer; because of the non-smooth term this minimizer is
characterized by a variational inequality (not a single linear system) and is
computed by the splitting scheme of Section~\ref{sec:blockbeta}.
\end{proposition}

\begin{proof}
Since $(\bs\otimes\bmm)^{2}=\bs\,(\bmm\cdot\bs)\,\bmm^{\!\top}=\bm0$, the
plastic distortion inverts explicitly, and the determinant follows from the
identity $\det(\bI+\bm{a}\otimes\bm{c})=1+\bm{a}\cdot\bm{c}$:
\begin{gather*}
  (\Fp)^{-1}=\bI-\beta\,\bs\otimes\bmm,\qquad
  \Fe=\bF-\beta\,(\bF\bs)\otimes\bmm,\\
  J=\det\bF\,(1-\beta\,\bs\cdot\bmm)=\det\bF,
\end{gather*}
proving (i) and (ii). At fixed $\by$, by (ii) the volumetric terms $-2\ln J$
and $(J-1)^{2}$ are constant in $\beta$; the shape term
$\tfrac{\mu}{2}\|\bF-\beta(\bF\bs)\otimes\bmm\|^{2}$ is a strictly convex
quadratic in $\beta$, with curvature $\mu\|\bF\bs\|^{2}>0$; the gradient term
$\tfrac{c_2}{2}(\nabla\beta\cdot\bs)^{2}$ is a convex quadratic in
$\nabla\beta$; and the BV term $q_c|\nabla\beta\cdot\bs|$ is convex - a norm
composed with the linear map $\nabla\beta\mapsto\nabla\beta\cdot\bs$ - but
non-smooth. This proves (iii). The sum $I[\by,\cdot]$ is therefore strictly
convex and its minimizer unique; vanishing of the first variation gives a
\emph{variational inequality} (an inclusion involving the subdifferential of
the $|\cdot|$ term), solved by the splitting scheme of
Section~\ref{sec:blockbeta}.
\end{proof}

\noindent The $\beta$-independence of $J$ keeps the $\beta$-block a
\emph{convex} subproblem (Section~\ref{sec:blockbeta}). For the homogeneous
(gradient-free) shape energy alone, the minimizer takes the closed form
\begin{equation*}
  \beta^{\ast}=\frac{(\bF\bs)\cdot(\bF\bmm)}{\|\bF\bs\|^{2}},
\end{equation*}
which enters the condensed energy of Proposition~\ref{prop:condensed}.

\subsection{The condensed energy}\label{sec:condensed}
For the homogeneous state ($\nabla\beta=0$, so the dislocation term drops out) the
condensed energy is obtained by minimizing the elastic energy over the single
scalar $\beta$ at fixed $\bF$; we compute it along the simple-shear path. The
mechanism analyzed in Section~\ref{sec:relax} - loss of quasiconvexity forcing
microstructure - is the one identified for this class of crystal-plasticity models
\citep{koster2015aformation}; what is new here is the closed-form
condensed energy of the Ciarlet--Geymonat crystal and its algebraic economy
(Remark~\ref{rem:degree}).

\begin{proposition}[Condensed energy in simple shear]\label{prop:condensed}
Let the homogeneous deformation be the plane-strain simple shear
$\bF=\bI+\gamma\,\bm{e}_1\otimes\bm{e}_2$ and take the representative slip
orientation $\varphi=-\pi/4$, i.e.\
$\bs=\tfrac1{\sqrt2}(1,-1,0)$, $\bmm=\tfrac1{\sqrt2}(1,1,0)$. Minimizing the
homogeneous elastic energy over the slip $\beta$ - letting the crystal slip
optimally at each imposed shear $\gamma$ - yields the \emph{condensed energy}
(the residual elastic energy that slip alone cannot remove)
\begin{equation}\label{eq:condensed}
  e(\gamma)=\frac{\mu\,\gamma^{2}(\gamma-2)^{2}}{4\,(\gamma^{2}-2\gamma+2)} .
\end{equation}
\end{proposition}

\begin{proof}[Proof of Proposition~\textup{\ref{prop:condensed}}]
Since the shear is isochoric, $\det\bF=1$ and hence $J=\det\Fe=1$ by
\eqref{eq:J}: the volumetric terms $-2\ln J$ and $(J-1)^{2}$ of \eqref{eq:psi}
vanish, and the homogeneous elastic energy reduces to $\tfrac{\mu}{2}(I_1-3)$
with $I_1=\tr\Ce$. From $\bF=\bI+\gamma\,\bm{e}_1\otimes\bm{e}_2$ the total
Cauchy-Green tensor is
\begin{equation}\label{eq:C-shear}
  \bC=\bF^{\!\top}\bF=
  \begin{pmatrix}1&\gamma&0\\[2pt]\gamma&1+\gamma^{2}&0\\[2pt]0&0&1\end{pmatrix},
\end{equation}
and contracting with $\bs=\tfrac1{\sqrt2}(1,-1,0)$ and
$\bmm=\tfrac1{\sqrt2}(1,1,0)$ gives the three scalars
\begin{equation}\label{eq:three-scalars}
  \bs\cdot\bC\bs=\frac{\gamma^{2}-2\gamma+2}{2},\qquad
  \bs\cdot\bC\bmm=-\frac{\gamma^{2}}{2},\qquad
  \tr\bC=3+\gamma^{2}.
\end{equation}
Expanding $\tr\Ce$ from $\Fe=\bF(\bI-\beta\,\bs\otimes\bmm)$ and
$\bs\cdot\bmm=0$,
\begin{equation*}
  \tr\Ce=\tr\bC-2\beta\,(\bs\cdot\bC\bmm)+\beta^{2}(\bs\cdot\bC\bs),
\end{equation*}
a strictly convex quadratic in $\beta$ (its leading coefficient
$\bs\cdot\bC\bs$ is positive). Its minimizer and minimum value are
\begin{equation}\label{eq:beta-star-shear}
  \beta^{\ast}=\frac{\bs\cdot\bC\bmm}{\bs\cdot\bC\bs}
  =-\frac{\gamma^{2}}{\gamma^{2}-2\gamma+2},\qquad
  \min_{\beta}\tr\Ce=\tr\bC-\frac{(\bs\cdot\bC\bmm)^{2}}{\bs\cdot\bC\bs}.
\end{equation}
Substituting \eqref{eq:three-scalars} into \eqref{eq:beta-star-shear},
\begin{equation}\label{eq:trCe-min}
\begin{split}
  \min_{\beta}\tr\Ce-3
  &=\gamma^{2}-\frac{\gamma^{4}/4}{(\gamma^{2}-2\gamma+2)/2}
  =\frac{2\gamma^{2}(\gamma^{2}-2\gamma+2)-\gamma^{4}}{2(\gamma^{2}-2\gamma+2)}\\
  &=\frac{\gamma^{2}(\gamma-2)^{2}}{2(\gamma^{2}-2\gamma+2)},
\end{split}
\end{equation}
where the last step uses
$2\gamma^{2}(\gamma^{2}-2\gamma+2)-\gamma^{4}
=\gamma^{4}-4\gamma^{3}+4\gamma^{2}=\gamma^{2}(\gamma-2)^{2}$. The
condensed energy is $e(\gamma)=\tfrac{\mu}{2}\bigl(\min_{\beta}\tr\Ce-3\bigr)$,
which is \eqref{eq:condensed}. 
\end{proof}

\section{Relaxation: lamellar microstructure}\label{sec:relax}
Recall \citep[see][]{dacorogna2008direct} that a continuous $W$ is \emph{quasiconvex} at $\bF$ if
$\int_{D}W(\bF+\nabla\bm\varphi)\,d\bx\ge|D|\,W(\bF)$ for every
$\bm\varphi\in C^{\infty}_{c}(D;\R^{3})$; quasiconvexity of the (condensed)
integrand is equivalent to weak lower semicontinuity of the functional. We show
that the condensed energy \eqref{eq:condensed} fails this test, so that
minimizing sequences develop the lamellar microstructure.

\begin{proposition}[Double well and non-quasiconvexity]\label{prop:nonqc}
The condensed energy \eqref{eq:condensed} is a double well: $e(0)=e(2)=0$,
$e(1)=\mu/4>0$, and $e''(1)=-\tfrac{3\mu}{2}<0$. Its two wells $\gamma=0$ and
$\gamma=2$ are rank-one connected and carry equal (minimal) energy; the simple
laminate mixing them attains the mean shear $\gamma$ at zero energy. Hence for
every $\gamma\in(0,2)$ the homogeneous state is not a minimizer, and the
condensed functional is not weakly lower semicontinuous - i.e.\ the integrand
is not quasiconvex in the sense of \citet{dacorogna2008direct}. Physically, the crystal
reaches the same imposed shear at strictly lower energy by alternating between two
equally favorable, geometrically compatible shear states - the origin of the
laminated grain pattern.
\end{proposition}

\begin{proof}
Direct substitution gives $e(0)=e(2)=0$ and $e(1)=\mu/4>0$. Writing
$e=\tfrac{\mu}{4}f/g$ with $f=\gamma^{2}(\gamma-2)^{2}$ and
$g=\gamma^{2}-2\gamma+2$, and using $f'(1)=g'(1)=0$, $f''(1)=-4$, $g''(1)=2$,
\begin{equation}\label{eq:eprime2}
  e''(1)=\frac{\mu}{4}\Bigl(\frac{f''(1)}{g(1)}-\frac{f(1)\,g''(1)}{g(1)^{2}}\Bigr)
  =\frac{\mu}{4}(-4-2)=-\frac{3\mu}{2}<0 ,
\end{equation}
so $e$ has strictly positive values between two equal-height zeros at
$\gamma=0,2$.

Next we show that the two wells are rank-one connected. They correspond to
$\bF_{\!A}=\bI$ ($\gamma=0$) and
$\bF_{\!B}=\bI+2\,\bm{e}_1\otimes\bm{e}_2$ ($\gamma=2$), whose difference
\begin{equation}\label{eq:rank-one}
  \bF_{\!B}-\bF_{\!A}=2\,\bm{e}_1\otimes\bm{e}_2
\end{equation}
is a rank-one tensor $\bm{a}\otimes\bm{n}$ with $\bm{a}=2\bm{e}_1$,
$\bm{n}=\bm{e}_2$. This is the Hadamard compatibility condition, so $\bF_{\!A}$
and $\bF_{\!B}$ can be joined across a planar interface of normal $\bm{n}$ by a
continuous, piecewise-affine deformation.

Finally, such a laminate beats the homogeneous state. For $\gamma\in(0,2)$ set
$\theta=\gamma/2\in(0,1)$ and form the layered field that
equals $\bF_{\!A}$ on a volume fraction $1-\theta$ and $\bF_{\!B}$ on $\theta$,
stacked along $\bm{n}=\bm{e}_2$. Its mean deformation gradient is
\begin{equation}\label{eq:laminate-mean}
  (1-\theta)\bF_{\!A}+\theta\bF_{\!B}
  =\bI+2\theta\,\bm{e}_1\otimes\bm{e}_2
  =\bI+\gamma\,\bm{e}_1\otimes\bm{e}_2 ,
\end{equation}
i.e.\ it satisfies the same average (hence the same boundary data
\eqref{eq:bc}), while each layer sits in a well, so its elastic energy is
$(1-\theta)e(0)+\theta\,e(2)=0$. Since $e(\gamma)>0$ by Step~1, the laminate has
strictly smaller energy than the homogeneous state. Therefore the infimum over
admissible fields lies below the homogeneous value $e(\gamma)$: the homogeneous state is not a minimizer and the condensed functional is not weakly lower semicontinuous, i.e.\ its integrand is not quasiconvex \citep{dacorogna2008direct}. 
\end{proof}

\begin{remark}[Grains as uniform slip combined with rigid rotation]\label{rem:grains}
The laminate admits a direct physical reading. In the $A$ layers
$\beta=0$ and $\Fe=\bI$: the lattice is undeformed. In the $B$ layers the
slip is uniform, $\beta=\beta_B$, and the elastic distortion is a
\emph{pure rotation} through the misorientation angle: at the well,
$\Fe=\bF_B(\bI-\beta_B\,\bs\otimes\bmm)$ satisfies
$(\Fe)^{\!\top}\Fe=\bI$ - for the representative $\varphi=-\pi/4$,
$\beta_B=-2$, explicitly
$\Fe=\bigl(\begin{smallmatrix}0&1\\-1&0\end{smallmatrix}\bigr)$, the
rotation through $|\theta|=\pi/2$. Each newly formed grain is therefore obtained
from the parent crystal by a nearly uniform plastic slip combined with a
rigid-body rotation of the lattice; both leave the lattice locally
stress-free, all incompatibility is pushed into the interfaces, and the
misorientation $\theta=2\varphi+\pi$ between neighboring grains is carried
entirely by the geometrically necessary boundaries (Section~\ref{sec:GNB}). Since $\beta_B=-\gamma_B$ for every $\varphi$, the laminate with fraction $\gamma/\gamma_B$ carries the mean slip $\langle \beta\rangle=-\gamma $ throughout the window.
\end{remark}

\begin{remark}[Algebraic economy of the CG condensed energy]\label{rem:degree}
Because the CG shape term $\tfrac{\mu}{2}(I_1-3)$ is \emph{linear} in $\Ce$, it
is only quadratic in $\beta$; eliminating $\beta$ leaves the degree-four rational
\eqref{eq:condensed}. The same linearity in $\Ce$ - there in the incompressible
neo-Hookean setting - underlies the analytical tractability of the laminate
relaxation of \citet{kochmann2011the}. The Saint-Venant--Kirchhoff energy
$\mu\,\tr((\bE^{e})^{2})$, $\bE^{e}=\tfrac12(\Ce-\bI)$, is quadratic in $\Ce$,
hence quartic in $\beta$, and condenses to the degree-eight rational of
\citet[Eq.~(31)]{koster2015aformation},
\[
  e_{\mathrm{SVK}}(\gamma)=\frac{\mu\,\gamma^{2}(\gamma-2)^{2}
  \bigl(\gamma^{4}-4\gamma^{3}+8\gamma^{2}-8\gamma+8\bigr)}
  {16\,(\gamma^{2}-2\gamma+2)^{2}} .
\]
Both energies share the same wells $\gamma\in\{0,2\}$ and the same optimal slip
$\beta^{\ast}=-\gamma^{2}/(\gamma^{2}-2\gamma+2)$, yet the CG well is
algebraically simpler - half the degree, with a lower barrier
$e(1)=\mu/4$ against $e_{\mathrm{SVK}}(1)=5\mu/16$. 
\end{remark}

\section{Geometrically necessary boundaries: finite thickness}\label{sec:GNB}
The relaxation of Section~\ref{sec:relax} produces, in the bare condensed model,
an infinitely fine laminate with sharp interfaces. The dislocation-density
gradient term restores a finite length scale, as follows.

\begin{remark}[The dislocation gradient sets the grain size]\label{rem:micro}
Two regimes must be distinguished. \emph{Without} the dislocation gradient term -
that is, for the bare condensed functional $\int e(\gamma)\,d\bx$ - the infimum is
not attained: minimizing sequences form ever finer laminates between the wells
$\gamma=0$ and $\gamma=2$, the layer spacing collapses to zero, and \emph{no grain size is defined}. It is precisely the gradient term
$\tfrac{c_2}{2}(\nabla\beta\cdot\bs)^{2}+q_c|\nabla\beta\cdot\bs|$, which carries
the dislocation density $\rho=\tfrac1b|\nabla\beta\cdot\bs|$, that penalizes the
transition layers, arrests the refinement, and selects a \emph{finite} wall
thickness - the geometrically necessary boundaries \citep{koster2015bformation}. The grain pattern is thus set by a competition: the double well drives mixing toward
finer layers, while the gradient penalty favors fewer, sharper walls, and the
balance fixes the wall thickness and the number of grains (the thickness in
\eqref{eq:hthickness} below, the grain count in Section~\ref{sec:method}). 
\end{remark}

The energy stored in a boundary, per unit boundary length, is the sharp-wall
value
\begin{equation}\label{eq:gammaG}
  \gamma_G=q_c\,|\beta_B| ,\qquad
  \beta_B=2\cot\varphi=-2\tan(\theta/2),
\end{equation}
where $\beta_B$ is the slip at the second well $\gamma_B=-2\cot\varphi$, at which
the elastic distortion reduces to a pure rotation \citep{koster2015aformation};
$\gamma_G$ is thus a closed-form quantity of the model. For small misorientations
$|\beta_B|\approx\theta$, so the boundary energy grows linearly,
$\gamma_G\approx\mu k\ell\,\theta$. The lattice
misorientation across the boundary is $\theta=2\varphi+\pi$; in the degenerate
orientation $\varphi=-\pi/4$ the rotation $\theta=\pi/2$ maps the cubic lattice
onto itself, so the boundary is coherent, carrying zero effective
misorientation, while for general $\varphi$ the boundary is incoherent.
Both $\gamma_G$ and the dislocation content are computed as functions of $\theta$
in Section~\ref{sec:method}.

Within a boundary layer the slip varies only across the layer, $\beta=\beta(x_2)$,
and the deformation gradient follows the rank-one (Hadamard) path
$\bF=\bI+\bm{a}(x_2)\otimes\bm{e}_2$ connecting the two wells. Since the elastic
energy contains no gradient of $\bm{a}$, the layer potential is obtained by
minimizing pointwise over $\bm{a}$,
\begin{equation}\label{eq:pdef}
  p(\beta)=\min_{\bm{a}\in\R^2}\ \frac{1}{\mu}\,
  \psi_{\mathrm{el}}\bigl((\bI+\bm{a}\otimes\bm{e}_2)(\bI-\beta\,\bs\otimes\bmm)\bigr).
\end{equation}
For the Ciarlet--Geymonat energy ($\lambda=0$) both minimizations are explicit -
the $a_1$-problem is a quadratic whose discriminant collapses to unity, and the
$a_2$-problem gives $(1+a_2)^2=1/Q$ - so the potential takes the \emph{closed
form}
\begin{equation}\label{eq:pclosed}
  p(\beta)=\frac12\Bigl[\frac1Q-1+\ln Q\Bigr],\qquad
  Q(\beta)=1+\beta\sin\varphi\,\bigl(\beta\sin\varphi-2\cos\varphi\bigr).
\end{equation}
Since $Q-1$ vanishes exactly at
$\beta=0$ and $\beta=\beta_B=2\cot\varphi$, and $p>0$ for $Q\neq1$, the potential
vanishes precisely at the two wells - the algebraic economy of the
Ciarlet--Geymonat crystal once more: for the Saint-Venant--Kirchhoff energy the
corresponding potential must be computed numerically. 

The Euler-Lagrange equation of the layer energy
$\int\bigl[\,\mu\,p(\beta)+\tfrac12 c_2\sin^2\!\varphi\,\beta'^2
+q_c|\sin\varphi|\,|\beta'|\,\bigr]dx_2$ admits the first integral
\begin{equation}\label{eq:firstint}
  \tfrac12\,k\ell^2\sin^2\!\varphi\,\beta'^2=p(\beta),
\end{equation}
the BV term dropping out of the Beltrami identity (its two
contributions cancel), so the boundary profile and thickness are set by the
gradient term alone. Separating variables and adopting the discreteness argument
of \citet{koster2015bformation} - the profile is cut off at the smallest slip quantum
$\beta_q=b/(L|\sin\varphi|)$ carried by a single dislocation - the layer
thickness is
\begin{equation}\label{eq:hthickness}
  h=\frac{\sqrt{k}\,\ell\,|\sin\varphi|}{\sqrt2}
  \int_{\beta_B+\beta_q}^{-\beta_q}\frac{d\beta}{\sqrt{p(\beta)}}\,.
\end{equation}
Table~\ref{tab:hthickness} evaluates \eqref{eq:hthickness} with the material
parameters of \citet{koster2015bformation} ($k=10^{-6}$, $\ell=400$~nm, and their
specimen ratio $b/L=10^{-4}$),
together with the same computation for the Saint-Venant--Kirchhoff potential,
evaluated numerically. The SVK column closely reproduces the published range of
$50$-to-$6$~nm for $\theta>10^\circ$, which validates the
pipeline; the residual difference reflects the cut-off details and the
$\lambda=0$ simplification adopted here. The Ciarlet--Geymonat walls are at most
$3\%$ thicker - the lower condensed barrier ($\mu/4$ against $5\mu/16$) widens
the profile only marginally - so the grain-boundary thickness predicted by the
theory is robust under the change of elastic energy.

\begin{table}[h]\centering\small
\begin{tabular}{@{}cccc@{}}
\toprule
$\theta$ (deg) & $h_{\mathrm{CG}}$ (nm) & $h_{\mathrm{SVK}}$ (nm) & ratio\\
\midrule
10 & 48.4 & 48.3 & 1.001\\
20 & 26.4 & 26.4 & 1.002\\
30 & 18.5 & 18.4 & 1.005\\
40 & 14.3 & 14.2 & 1.009\\
50 & 11.8 & 11.6 & 1.013\\
60 & 10.0 &  9.9 & 1.018\\
70 &  8.8 &  8.6 & 1.024\\
80 &  7.8 &  7.6 & 1.029\\
\bottomrule
\end{tabular}
\caption{Boundary-layer thickness $h(\theta)$ from \eqref{eq:hthickness} with the
material parameters of \citet{koster2015bformation} ($k=10^{-6}$, $\ell=400$~nm,
$b/L=10^{-4}$): the Ciarlet--Geymonat closed form \eqref{eq:pclosed} against the
numerically evaluated Saint-Venant--Kirchhoff potential. The SVK column closely
reproduces the published $50$-to-$6$~nm range; the Ciarlet--Geymonat walls are at
most $3\%$ thicker. }
\label{tab:hthickness}
\end{table}

\begin{remark}[Phase-field analogy]
The regularized functional has the structure of a phase-field model of
Cahn--Hilliard--Modica--Mortola type \cite{cahn1958free,modica1987gradient}: the slip $\beta$ acts as a non-conserved order parameter whose preferred values $\beta=0$ and $\beta=\beta_B$ label the two variants; the condensed energy supplies the double-well bulk term; and the dislocation energy
$q_c|\nabla\beta\cdot\bs| + \tfrac{c_2}{2}(\nabla\beta\cdot\bs)^2$ plays the
role of the interfacial energy. The first integral \eqref{eq:firstint} is the equipartition
of bulk and gradient energy across the interface -- the optimal-profile
construction of phase-field theory -- and \eqref{eq:hthickness} the corresponding
interface-width formula. Two features distinguish the present model. First,
the interfacial energy is strongly anisotropic: only gradients along $\bs$
are penalized, so interfaces normal to $\bmm$ are energetically free -- the
slip laminates of Section~8. Second, it is physical rather than
phenomenological: $q_c=\mu k\ell$ and $c_2=\mu k\ell^2$ are the line and
interaction energies of the geometrically necessary dislocations, so the
boundary energy $\gamma_G$ and the wall thickness are predictions in terms of
dislocation parameters rather than fitted interface constants. In particular,
unlike phase-field models of grain boundaries in which the lattice
orientation is an independent order parameter with a postulated boundary
energy \cite{kobayashi2000continuum}, here the misorientation, the boundary energy and the thickness all emerge from the slip kinematics and the dislocation energetics.
\end{remark}

\section{Numerical solution: a block-coordinate finite element method}\label{sec:method}

\subsection{Nondimensionalization}\label{sec:nondim}
Recall $\rho/\rho_s=\ell\,|\nabla\beta\cdot\bs|$ with $\ell=1/(b\rho_s)$ from
\eqref{eq:disl-moduli}. Scaling positions by the specimen size $L$
($\tilde\bx=\bx/L$) and the energy by $\mu L^{2}$ (plane strain), the
functional becomes
\begin{equation}\label{eq:nondim}
  \frac{I}{\mu L^{2}}=\int_{\tilde V}\Bigl[
  \tfrac12(I_1-3-2\ln J)+\tfrac{\tilde\lambda}{2}(J-1)^{2}
  +k\eta\,|\tilde\nabla\beta\cdot\bs|
  +\tfrac{k\eta^{2}}{2}\,(\tilde\nabla\beta\cdot\bs)^{2}\Bigr]d\tilde\bx,
\end{equation}
with $\tilde\lambda=\lambda/\mu$ and $\eta=\ell/L$.
The dimensionless internal length $\eta=\ell/L$ sets the grain-boundary
thickness ($\eta\to0$ drives finer microstructure); the BV and gradient moduli
become $k\eta$ and $k\eta^{2}$. Henceforth we use these variables and drop the
tildes.

\subsection{Finite element discretization}\label{sec:fem}
We use a finite element discretization, which handles general geometries and
boundary conditions directly (spectral/FFT schemes, as in \citet{moulinec1998numerical},
are an efficient alternative on periodic cells). In plane strain the fields
depend on $(x_1,x_2)$. We mesh $V$ and interpolate both the deformation
$\by$ and the plastic slip $\beta$ by continuous Lagrange elements,
\begin{equation}\label{eq:fe-interp}
  \by\approx\bN_{\!y}\,\hat\by,\qquad \beta\approx\bN\bbeta,
\end{equation}
with $\bN$ the Lagrange shape functions - bilinear $Q_1$ quadrilaterals in the
computations of Section~\ref{sec:results} - and $\bN_{\!y}$ their
vector-valued counterpart; gradients use the isoparametric map through
$\nabla\bN$. The essential data are imposed nodally: $\by=\bF\bx$ and $\beta=0$
on $\partial V$. All integrals - the residual \eqref{eq:Y-weak}, the tangent
\eqref{eq:tangent-form}, and the $\beta$-system \eqref{eq:blockbeta} - are
evaluated by Gauss quadrature exact for the element order. The same $\bN$ enters
the $\beta$-block \eqref{eq:blockbeta}, so no separate interpolation is
introduced.

\subsection{The slip block: a convex non-smooth solve}\label{sec:blockbeta}
At a fixed deformation $\by$ - hence a fixed total Cauchy-Green tensor
$\bC=\bF^{\!\top}\bF$ - Proposition~\ref{prop:convex-beta} makes the energy in
$\beta$ convex: a strictly convex quadratic plus the non-smooth BV term. We
derive the resulting solve.

By Proposition~\ref{prop:convex-beta}(ii) the volumetric terms are constant in
$\beta$; only the shape and dislocation parts vary. Since
$\Ce=(\bI-\beta\,\bmm\otimes\bs)\,\bC\,(\bI-\beta\,\bs\otimes\bmm)$,
\begin{equation}\label{eq:trCe}
  \tr\Ce=\tr\bC-2\beta\,(\bs\cdot\bC\bmm)+\beta^{2}(\bs\cdot\bC\bs).
\end{equation}
With $d:=\nabla\beta\cdot\bs$ and the moduli \eqref{eq:disl-moduli}, the
$\beta$-dependent energy density is
\begin{equation}\label{eq:eb-density}
  \frac{\mu}{2}\bigl[\beta^{2}(\bs\cdot\bC\bs)-2\beta(\bs\cdot\bC\bmm)\bigr]
  +q_c\,|d|+\frac{c_2}{2}\,d^{2}.
\end{equation}
Because of the non-smooth $|d|$, stationarity of \eqref{eq:functional} in $\beta$
($\beta=0$ on $\partial V$) is the inclusion: for all admissible $\delta\beta$,
\begin{equation}\label{eq:eb-weak}
  \int_{V}\!\Bigl[\mu(\bs\cdot\bC\bs)\,\beta-\mu(\bs\cdot\bC\bmm)\Bigr]\delta\beta
  +\bigl[c_2\,d+q_c\,\xi\bigr](\nabla\delta\beta\cdot\bs)\,d\bx=0,
\end{equation}
with $\xi\in\partial|d|$, i.e.\ $\xi=\operatorname{sign}(d)$ for $d\neq0$ and
$\xi\in[-1,1]$ where $d=0$.
Equation \eqref{eq:eb-weak} carries a direct physical meaning: its strong form is
the balance
\begin{equation}\label{eq:schmid}
  \tau_r=\varsigma,\qquad
  \tau_r:=\mu\bigl[\bs\cdot\bC\bmm-\beta\,(\bs\cdot\bC\bs)\bigr],\qquad
  \varsigma:=-\nabla\cdot\bigl[(c_2\,d+q_c\,\xi)\,\bs\bigr],
\end{equation}
between the resolved shear stress
$\tau_r=-\partial\psi_{\mathrm{el}}/\partial\beta$ on the slip system and the
back stress $\varsigma$ generated by the geometrically necessary dislocations -
the counterpart of the Schmid stress and back stress of
\citet{koster2015aformation}. The non-smooth term acts as a threshold of
rate-independent type: where the slip gradient vanishes, $\xi$ adjusts within
$[-1,1]$, and dislocation walls form only where the resolved stress can no
longer be balanced below the line-tension threshold set by $q_c$.
To solve this inclusion we introduce the auxiliary field
$d\approx\nabla\beta\cdot\bs$ and split by the alternating direction method of
multipliers (ADMM). With penalty $r>0$ and dual $\bm{y}_d$, one alternates:
\begin{itemize}\setlength{\itemsep}{3pt}\setlength{\topsep}{3pt}
  \item[(a)] \emph{$\beta$-update (SPD linear solve).} Interpolating
  $\beta=\bN\bbeta$ and assembling the smooth part,
  \begin{equation}\label{eq:blockbeta}
  \begin{split}
    \Bigl[\,\mu\,(\bs\cdot\bC\bs)\,&\bN^{\!\top}\bN
    +(c_2{+}r)(\nabla\bN\cdot\bs)^{\!\top}(\nabla\bN\cdot\bs)\Bigr]\bbeta\\
    &=\mu\,(\bs\cdot\bC\bmm)\,\bN^{\!\top}
    +(\nabla\bN\cdot\bs)^{\!\top}(r\,d-\bm{y}_d);
  \end{split}
  \end{equation}
  \item[(b)] \emph{$d$-update (soft threshold).}
  $d\leftarrow\mathcal{S}_{q_c/r}\bigl(\nabla\beta\cdot\bs+\bm{y}_d/r\bigr)$,
  with $\mathcal{S}_\kappa(z)=\operatorname{sign}(z)\max(|z|-\kappa,0)$;
  \item[(c)] \emph{dual update.}
  $\bm{y}_d\leftarrow\bm{y}_d+r\,(\nabla\beta\cdot\bs-d)$.
\end{itemize}
The matrix in (a) is symmetric positive definite
(Proposition~\ref{prop:convex-beta}); each sweep is one SPD solve plus a
pointwise threshold, and iterating solves the convex $\beta$-block. In the limit
$q_c\to0$ the BV term drops and (a) alone becomes the single linear solve.

\subsection{The deformation block: Levenberg-regularized Newton}\label{sec:blocky}
At a fixed plastic slip $\beta$ the dislocation term of \eqref{eq:psi} depends on
$\nabla\beta$ alone and is therefore constant; the $\by$-block minimizes the
elastic energy
\begin{equation}\label{eq:Iel}
  I_{\mathrm{el}}[\by]=\int_{V}\psi_{\mathrm{el}}(\Fe)\,d\bx,\qquad
  \Fe=\nabla\by\,(\Fp)^{-1},\quad (\Fp)^{-1}=\bI-\beta\,\bs\otimes\bmm,
\end{equation}
over $\by$ with $\by=\bF\bx$ on $\partial V$. As $\psi_{\mathrm{el}}$ is
polyconvex but \emph{not} convex (Remark~\ref{rem:notconvex}), $I_{\mathrm{el}}$
is nonconvex in $\by$; we therefore globalize Newton's method by a Levenberg
regularization.

The linearization uses the standard identities
$\partial I_1/\partial\Fe=2\Fe$, Jacobi's formula
$\partial J/\partial\Fe=J(\Fe)^{-\top}$ (hence
$\partial(\ln J)/\partial\Fe=(\Fe)^{-\top}$) and the variation of the inverse
transpose
$\delta[(\Fe)^{-\top}]=-(\Fe)^{-\top}(\delta\Fe)^{\!\top}(\Fe)^{-\top}$,
together with the pull-back
$\delta\Fe=\nabla(\delta\by)\,(\Fp)^{-1}$, which follows from \eqref{eq:Iel} at
fixed $\beta$. Differentiating
$\psi_{\mathrm{el}}=\frac{\mu}{2}(I_1-3-2\ln J)+\frac{\lambda}{2}(J-1)^2$
yields the first Piola-Kirchhoff stress
\begin{equation}\label{eq:Pe}
\begin{split}
  \Pe:=\frac{\partial\psi_{\mathrm{el}}}{\partial\Fe}
  &=\mu\bigl(\Fe-(\Fe)^{-\top}\bigr)+\lambda J(J-1)(\Fe)^{-\top}\\
  &=\mu\Fe+\bigl[\lambda J(J-1)-\mu\bigr](\Fe)^{-\top}.
\end{split}
\end{equation}
In particular $\Pe=\bm0$ at $\Fe=\bI$, so the reference state is stress-free.
By the pull-back identity, stationarity of \eqref{eq:Iel} reads, for all
$\delta\by$ vanishing on $\partial V$,
\begin{equation}\label{eq:Y-weak}
  G(\by)[\delta\by]
  =\int_{V}\Pe:\bigl(\nabla(\delta\by)\,(\Fp)^{-1}\bigr)\,d\bx
  =\int_{V}\bigl(\Pe\,(\Fp)^{-\top}\bigr):\nabla(\delta\by)\,d\bx=0 .
\end{equation}
Differentiating once more, the elastic tangent
$\mathbb{A}=\partial^{2}\psi_{\mathrm{el}}/\partial\Fe\,\partial\Fe
=\partial\Pe/\partial\Fe$ has components (with $(\Fe)^{-1}$ the inverse)
\begin{equation}\label{eq:Amod}
\begin{split}
  \mathbb{A}_{iJkL}
  &=\mu\,\delta_{ik}\delta_{JL}
  +\lambda(2J-1)J\,(\Fe)^{-1}_{Ji}(\Fe)^{-1}_{Lk}\\
  &\quad+\bigl[\mu-\lambda J(J-1)\bigr](\Fe)^{-1}_{Jk}(\Fe)^{-1}_{Li},
\end{split}
\end{equation}
and, by the pull-back identity, the tangent bilinear form in $\by$ is
\begin{equation}\label{eq:tangent-form}
  DG(\by)[\delta\by,\bm{w}]
  =\int_{V}\bigl(\nabla\bm{w}\,(\Fp)^{-1}\bigr):\mathbb{A}:
  \bigl(\nabla(\delta\by)\,(\Fp)^{-1}\bigr)\,d\bx .
\end{equation}
Interpolating $\by=\bN_{\!y}\hat{\by}$ and assembling \eqref{eq:Y-weak} and
\eqref{eq:tangent-form} gives the residual vector $\bm{g}(\hat{\by})$ and the
tangent stiffness $\bK(\hat{\by})$. Because $\bK$ may be indefinite where
$\psi_{\mathrm{el}}$ is nonconvex, we take the Levenberg-regularized step
\begin{equation}\label{eq:LM}
  \bigl(\bK(\hat{\by})+\lambda_{\mathrm{LM}}\bM\bigr)\,\delta\hat{\by}
  =-\,\bm{g}(\hat{\by}),\qquad \bM\ \text{SPD},\ \ \lambda_{\mathrm{LM}}\ge0,
\end{equation}
with $\bM$ a fixed symmetric positive-definite metric (mass matrix). The scalar
$\lambda_{\mathrm{LM}}$ is chosen adaptively - increased until the update
lowers $I_{\mathrm{el}}$, decreased after each successful step. As
$\lambda_{\mathrm{LM}}\!\to\!0$, \eqref{eq:LM} recovers the quadratically
convergent Newton step; as $\lambda_{\mathrm{LM}}\!\to\!\infty$ it tends to the
preconditioned gradient step $-\lambda_{\mathrm{LM}}^{-1}\bM^{-1}\bm{g}$,
guaranteeing energy descent even in the nonconvex regime.

\subsection{Detecting the onset of microstructure}\label{sec:bifurcation}
Along the loading path parametrized by $\gamma$, the homogeneous state is a
critical point of the alternating map; because the functional is
non-quasiconvex (Proposition~\ref{prop:nonqc}), this homogeneous state ceases to
be a minimizer at some $\gamma_c$ beyond which a laminate has lower energy. Because the homogeneous state remains a fixed point of the alternating iteration, the laminate is reached only by breaking its symmetry. We use an energy-based perturbation test.

\smallskip
\noindent\textbf{Perturbation test.}
At each converged state $(\by,\beta)$, of nondimensional total energy
$E:=I[\by,\beta]$, we perturb the slip on the interior nodes,
$\beta\leftarrow\beta+\zeta$ with $\zeta$ random of amplitude $\varepsilon$, and
re-run the alternating solver to a new converged state $(\by',\beta')$ of energy
$E'$. The perturbed state is accepted when it lowers the energy, and a genuine
symmetry-breaking bifurcation is recorded when the drop is significant \emph{and}
the slip field changes:
\begin{equation}\label{eq:bif-test}
  E-E'>c\,(1+|E|)\qquad\text{or}\qquad
  \frac{\|\beta'-\beta\|}{\|\beta\|}>\delta ,
\end{equation}
with small fixed tolerances $c,\delta>0$. A sub-threshold perturbation that does
not lower the energy is rejected and the
homogeneous branch is retained. Sweeping $\gamma$, the first accepted perturbation brackets the onset $\gamma_c$.

\begin{remark}[Stability interpretation]\label{rem:schur}
The exact stability boundary is where the reduced Hessian (Schur complement)
$\bm{S}=\bK_{yy}-\bK_{y\beta}\bK_{\beta\beta}^{-1}\bK_{\beta y}$ loses positivity
($\bK_{\beta\beta}\succ0$ by Proposition~\ref{prop:convex-beta}). The
perturbation test \eqref{eq:bif-test} probes this boundary without assembling
$\bm{S}$: a perturbation can lower the energy precisely when $\bm{S}$ has a
non-positive direction. Monitoring the inertia of $\bm{S}$ directly - an
integer count from a Bunch-Kaufman factorization, robust to round-off and
certifying a \emph{simple} crossing - is a sharper alternative.
\end{remark}

\subsection{The alternating scheme}\label{sec:alternating}
The two blocks are combined into a block-coordinate (alternating) minimization:
the convex $\beta$-block by ADMM (Section~\ref{sec:blockbeta}) and the
$\by$-block by Levenberg-regularized Newton (Section~\ref{sec:blocky}).

\smallskip
\noindent\fbox{\parbox{0.96\linewidth}{%
\textbf{Algorithm (one load step, prescribed $\bF$).}
\begin{enumerate}\setlength{\itemsep}{2pt}\setlength{\topsep}{3pt}
  \item Carry ($\by, \beta)$ from the previous load step; affine predictor for $\by$.
  \item \textbf{repeat}
    \begin{enumerate}\setlength{\itemsep}{1pt}
      \item[(a)] \emph{$\beta$-update:} solve the convex block by the ADMM
        sweep (a)-(c) of Section~\ref{sec:blockbeta} for
        $\beta\leftarrow\argmin_\beta I[\by,\beta]$;
      \item[(b)] \emph{$\by$-update:} take Levenberg steps \eqref{eq:LM} to
        decrease $I[\by,\beta]$;
    \end{enumerate}
    \textbf{until} $|\Delta I|<\mathrm{tol}$.
  \item Apply the perturbation test \eqref{eq:bif-test} to check for onset of
    microstructure.
\end{enumerate}}}

\begin{proposition}[Monotone descent and convergence]\label{prop:descent}
The energy sequence $I_n:=I[\by^{n},\beta^{n}]$ produced by the alternating
scheme is non-increasing and convergent.
\end{proposition}

\begin{proof}
As $\beta^{n+1}=\argmin_\beta I[\by^{n},\beta]$ is the exact (unique) minimizer
(Proposition~\ref{prop:convex-beta}), and each accepted Levenberg step
\eqref{eq:LM} lowers the elastic energy (Section~\ref{sec:blocky}),
\begin{equation*}
  I[\by^{n+1},\beta^{n+1}]\ \le\ I[\by^{n},\beta^{n+1}]\ \le\
  I[\by^{n},\beta^{n}] ,
\end{equation*}
so the sequence $I_n$ is non-increasing. By Lemma~\ref{lem:coercive} the
energy is bounded below, and a non-increasing sequence bounded below converges.
\end{proof}

\begin{remark}\label{rem:cvg}
Any limit point is a critical point: the $\beta$-block sits at its unique minimizer and
the $\by$-block at a Newton fixed point. Global optimality is not guaranteed -
consistent with the non-quasiconvexity (Proposition~\ref{prop:nonqc}), which is
exactly why microstructure, flagged by the perturbation test
\eqref{eq:bif-test}, emerges beyond $\gamma_c$.
\end{remark}

\subsection{Numerical results}\label{sec:results}
We apply the scheme to plane-strain simple shear of the unit cell. The slip
orientation is $\varphi=-1.4$ (misorientation $2\varphi+\pi\approx19.6^\circ$),
the second Lam\'e constant is $\lambda=0$, and the dislocation modulus and
internal length are $k=8\times10^{-5}$ and $\eta=\ell/L=1.9\times10^{-3}$ -
chosen so that several layers fit inside the cell for shears in the
laminate window $0<\gamma<\gamma_B$, $\gamma_B=-2\cot\varphi\approx0.345$. The overall shear $\gamma$ is raised in sixteen equal increments of $0.0206$, from $0.031$ to $0.34$ on a uniform $Q_1$ mesh; at each step the alternating scheme
(Section~\ref{sec:alternating}) is run to convergence and the perturbation test
\eqref{eq:bif-test} is applied. Other chosen parameters are: (i) perturbation amplitude $\varepsilon=10^{-3}$, (ii) acceptance by strict energy decrease; flagging thresholds of \eqref{eq:bif-test}, $c=10^{-7}$, $\delta = 10^{-2}$, (iii) staggered and Newton tolerances $10^{-10}$ and $10^{-9}$, respectively, with $\mathrm{rng}(1)$ fixed, so all runs are deterministic.

\paragraph{Laminate formation}
Figure~\ref{fig:res-lam} shows the converged fields at $\gamma=0.216$. The plastic
slip $\beta$ organizes into layers nearly parallel to the shear direction; the
lattice rotation $\theta_e$ (the in-plane rotation angle in the polar
decomposition of $\Fe$) alternates between bands of distinct orientation -
the \emph{grains} - and the geometrically necessary dislocation density
$\rho/\rho_s$ concentrates on the interfaces between them - the geometrically
necessary boundaries. Thus a direct energy minimization with the polyconvex
Ciarlet--Geymonat energy reproduces the lamellar grain structure of
\citet{koster2015bformation}; laminate microstructures of precisely this kind are
observed experimentally in shear-deformed copper single crystals
\citep{dmitrieva2009lamination}. Within each lamella the
computed slip is nearly constant at a well value and the lattice rotation
$\theta_e$ nearly constant, so the grains are almost stress-free regions of
uniform slip combined with rigid-body lattice rotation
(Remark~\ref{rem:grains}), with the incompatibility concentrated in the
walls - the interpretation established in
\citep{koster2015bformation,koster2015aformation}, here recovered from an
unconstrained two-dimensional minimization rather than built into a
one-dimensional ansatz. The two grain families differ by the lattice misorientation
$\theta=2\varphi+\pi$ carried by the wells - incoherent boundaries for the
present $\varphi=-1.4$ ($\theta\approx19.6^\circ$; Section~\ref{sec:GNB}). At
equilibrium the interfaces lie nearly parallel to
the shear direction, as the theory predicts; the experimental observation
\citep{koster2015aformation} that grains are initially misaligned and become
elongated and aligned with the shear direction only after several deformation
passes reflects the approach of the microstructure to this equilibrium. A
closed-form trajectory of the boundary inclination with loading is left to future
work.

\begin{figure}[h]\centering
\includegraphics[width=\linewidth]{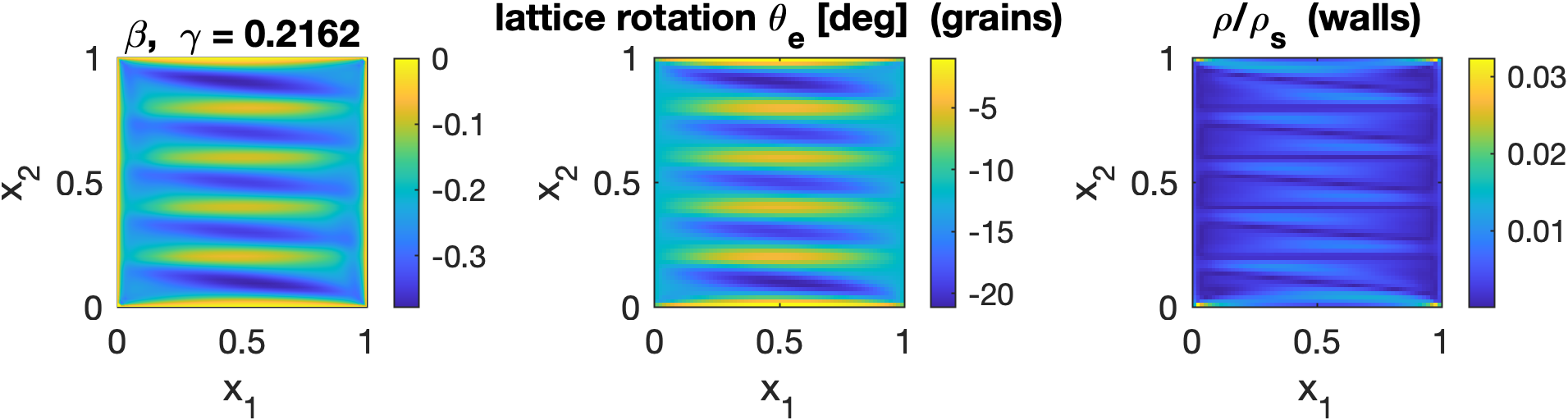}
\caption{Converged fields at $\gamma=0.216$ ($64\times64$ mesh, continuation protocol of the mesh study; four B-lamellae): plastic slip
$\beta$ (left), lattice rotation $\theta_e$ in degrees (middle, the grains), and
dislocation density $\rho/\rho_s$ (right, the walls). The layers are nearly
parallel to the shear axis.}
\label{fig:res-lam}
\end{figure}

\paragraph{Energy response}
Figure~\ref{fig:res-E}(left) plots $E(\gamma)$: it rises to a barrier near
$\gamma\approx0.19$ and then descends toward a second minimum near
$\gamma\approx0.32$ - the double-well shape of the condensed energy
(Proposition~\ref{prop:nonqc}). For the present orientation the analytic well of
\eqref{eq:condensed} sits at $\gamma_B\approx0.345$ (the representative
$\varphi=-\pi/4$ used to display \eqref{eq:condensed} places it at $\gamma=2$),
while the coupled computed response reaches its minimum slightly earlier, near
$\gamma\approx0.32$, the sweep stopping just below $\gamma_B$. As an internal
check, the volume-averaged first Piola stress $\tau$ (middle) vanishes at the
barrier and at the second minimum, matching $\tau=\mathrm{d}E/\mathrm{d}\gamma$ to graphical accuracy; the slip amplitude (right) grows toward $|\beta_B|$, the magnitude of the optimal slip $\beta_B=\beta^\ast$ at the second well $\gamma_B$, as the layers develop. A random perturbation of the homogeneous state lowers the energy already at the smallest sampled shear, so $\gamma_c\le 0.03$; that the homogeneous state is nowhere a minimizer in $0<\gamma<\gamma_B$ is guaranteed by Proposition 13.

\begin{figure}[h]\centering
\includegraphics[width=\linewidth]{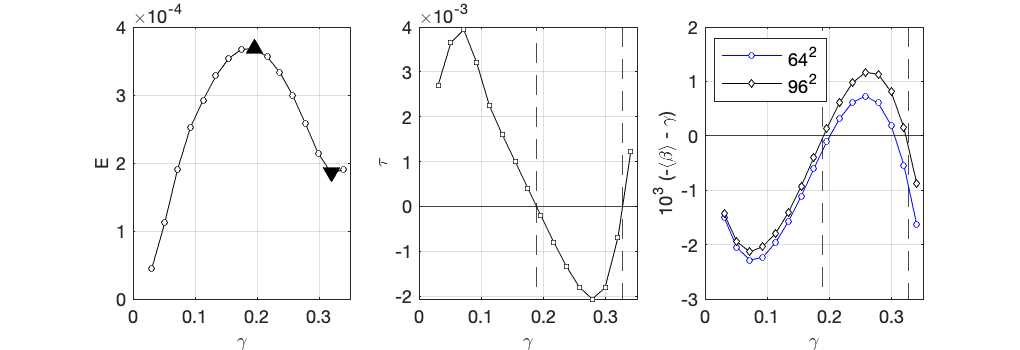}
\caption{Left: energy $E(\gamma)$ along the continuation branch in the hard
device, a double well with the maximum ($\blacktriangle$, $\gamma\approx0.19$)
and the second minimum ($\blacktriangledown$, $\gamma\approx0.32$). Middle:
averaged stress $\tau$; its sign changes fall in the same load increments as
the extrema of $E$. Right: deviation of the mean slip from the lamellar
relation $\langle\beta\rangle=-\gamma$ of Section~\ref{sec:relax} on the
$64^2$ and $96^2$ meshes; it stays below $2.3\times10^{-3}$ in magnitude and
changes sign together with $\tau$. (Left and middle panels: $96\times96$
mesh.)}
\label{fig:res-E}
\end{figure}

The computed branch consists of local minimizers of $I$ at fixed boundary
data and is therefore metastable in the hard device: the globally minimal
response follows the lower convex envelope of the energy, along which the
stress is constant at the Maxwell value \citep{le2019introduction}, and
neither the continuation nor the perturbation test~\eqref{eq:bif-test}
leaves the computed branch for those relaxed states. The double well of
Figure~\ref{fig:res-E} is thus a property of the branch selected by
quasi-static loading, with the pattern formed at small shear carried through
the sweep.

\paragraph{Mesh behavior and non-attainment}
Refining the mesh produces a \emph{finer} laminate (four layers at $64^2$, five at
$96^2$) and a lower energy (Table~\ref{tab:res-mesh}). This is not a defect of the
solver but the numerical signature of non-quasiconvexity
(Proposition~\ref{prop:nonqc}): because the internal length here lies well below
the mesh size, the gradient regularization is unresolved and the computation
reflects the sharp-interface limit, in which the infimum is not attained and finer
meshes capture deeper members of the minimizing sequence. The same mechanism -
numerical error governed by the energy of the microstructural interfaces - is
identified in the systematic comparison of \citet{kumar2020assessment}, where
finite element schemes for unregularized non-convex energies often fail to
converge at all; here the physical gradient regularization guarantees a
minimizer (Proposition~\ref{prop:existence}), and the mesh dependence merely
reflects that the internal length is not yet resolved. The energy therefore
decreases toward its relaxed limit as $h\to0$; on a single mesh neither the absolute energy nor the layer count is the invariant quantity.

\begin{table}[h]\centering\small
\begin{tabular}{@{}cccc@{}}
\toprule
$\gamma$ & $E$ ($64^2$) & $E$ ($96^2$) & relative change \\
\midrule
$0.196$ (barrier) & $4.61\times10^{-4}$ & $3.68\times10^{-4}$ & $20.1\%$ \\
$0.319$ (second well) & $2.96\times10^{-4}$ & $1.86\times10^{-4}$ & $37.2\%$ \\
\bottomrule
\end{tabular}
\caption{Mesh sensitivity of the total energy: the laminate refines and the
energy decreases toward the relaxed limit as $h\to0$, consistent with
non-attainment.}
\label{tab:res-mesh}
\end{table}

\paragraph{The slip field does not converge pointwise; its mean does}
Across an interface with unit normal $\bmm$, a jump $[\beta]$ of the slip
produces the plastic-distortion jump $[\beta]\,\bs\otimes\bmm$ -- a rank-one
connection, hence kinematically compatible and free of elastic energy; and
since the dislocation density is $|\nabla\beta\cdot\bs|$, gradients of
$\beta$ along $\bmm$ carry no dislocation energy either. Slip laminates
normal to $\bmm$ are therefore free in the functional -- the mechanism by
which the single-slip energy loses quasiconvexity
\citep{ortiz1999nonconvex,carstensen2002non,conti2005single,conti2005dislocation}
-- and the discrete minimizer exploits them at the finest scale available to it, namely the mesh (Table 3): at $\gamma=0.34$ the transverse gradient $\max|\nabla\beta\cdot\bmm|$ grows from $28.8$ ($64^2$) to $43.2$ ($96^2$), in the ratio $1.50$ of the mesh sizes, matching the wall gradient $\max|\nabla\beta\cdot\bs|$ ($29.9$ and $44.9$), and $\max|\beta|$ exceeds $|\beta_B|$ by $40\%$ on both meshes: the pointwise signature of a minimizing sequence for an infimum that is not
attained, and no more a defect of the computation than the refinement of the
layer count is. The averages, by contrast, converge: the area-weighted mean
slip obeys the lamellar relation $\langle\beta\rangle=-\gamma$ to within
$2.3\times10^{-3}$ along the entire loading path on both meshes, extrapolating at $\gamma=0.34$ to $0.340\pm 0.001$ (Table 3). Its deviation changes sign together with $\tau$ -- the averaged stress tracks the lag of the mean transformation behind the applied shear (right panel of Figure~\ref{fig:res-E}).

The unpenalized direction could be closed by a transverse term
$\varepsilon_m c_2(\nabla\beta\cdot\bmm)^2$. We do not add it: it would
assign energy to compatible, dislocation-free -- hence physically free --
structures, a change of model rather than a regularization, and none of the
quantities reported here requires it.

\begin{remark}[Multiple slip]
The free direction is a structural feature of single slip: by \eqref{eq:nye}, the
dislocation density of system $i$ involves only $\nabla\beta_i\cdot\bs_i$, so
gradients of each slip along its own plane normal remain dislocation-free for
any number of systems. What changes with several \emph{interacting} systems
is the room such perturbations have. The free direction is system-specific --
with slip directions spanning the plane, no single direction carries a free
oscillation of the total plastic distortion, and energetically free laminates
must be pure in one system; the interaction of the dislocation networks, as
in the double-slip continuum dislocation theory of \citet{le2008plane,le2009plane}, couples the densities wherever walls of different systems meet;
and latent-hardening or dissipative contributions penalize slip amplitude
irrespective of the gradient direction, lifting the degeneracy altogether. We
therefore expect the pointwise non-uniqueness documented above to be specific
to the energetic single-slip idealization, and the multi-slip problem to be
substantially less degenerate.
\end{remark}

\begin{table}[h]\centering\small
\caption{Pointwise and averaged measures of the slip field at $\gamma=0.34$.
Gradients in units of $1/L$; the mesh ratio is $1.50$.}
\label{tab:res-pointwise}
\begin{tabular}{lccc}
\hline
 & $64^2$ & $96^2$ & ratio \\
\hline
$\max|\nabla\beta\cdot\bmm|$ (transverse) & $28.8$ & $43.2$ & $1.50$ \\
$\max|\nabla\beta\cdot\bs|$ (wall)         & $29.9$ & $44.9$ & $1.50$ \\
$\max|\beta|\,/\,|\beta_B|$               & $1.40$ & $1.40$ & $1.00$ \\
$-\langle\beta\rangle\,/\,\gamma$          & $0.995$ & $0.997$ & \\
\hline
\end{tabular}
\end{table}

\paragraph{Relaxation and the relaxed energy}
The mesh dependence above is the numerical face of \emph{relaxation}. The two
wells of the condensed energy carry zero energy and are rank-one connected
(Proposition~\ref{prop:nonqc}), so the quasiconvex envelope - the relaxed
energy - is $e^{\ast\ast}(\gamma)=0$ on the whole laminate window $\gamma\in[0,\gamma_B]$.
Laminates mixing the wells drive the bulk energy to this infimum, while the
dislocation gradient term supplies a positive interfacial cost that grows as the
layers are refined; the total energy therefore approaches $e^{\ast\ast}$ only in
the joint limit of a vanishing internal length. This is exactly the
non-attainment reported for this class of energies
\citep{koster2015aformation}: the physically meaningful outputs are the
relaxed energy and the finite boundary quantities below, not the energy of any
single admissible field.

\paragraph{Local versus global minimizers} A distinction the numerics make
concrete is worth stating. The \emph{bare} condensed functional $\int
e(\gamma)\,d\bx$, without the gradient term, has no minimizer: its infimum is the
relaxed value $e^{\ast\ast}\equiv0$, approached but never attained by ever finer
laminates. The \emph{regularized} functional actually solved here does admit a
minimizer (Proposition~\ref{prop:existence}) - a laminate of finite period; but since the internal length is unresolved on the present meshes, that period is set by the mesh rather than by $\eta$ (Table~\ref{tab:res-mesh}), and the computed state tracks the sharp-interface infimum. A descent method alone would stall at the higher-energy \emph{homogeneous} critical point; the perturbation test exists precisely to
leave it, accepting any energy-lowering symmetry break so that the iteration
approaches the relaxed infimum rather than this local minimizer - the global
viewpoint favored there \citep{kumar2020assessment}.

\paragraph{Resolved internal length}
The computations above operate in the sharp-interface regime: the internal
length lies below the mesh size, and the laminate period is set by the mesh
(Table~\ref{tab:res-mesh}). To demonstrate that the gradient regularization
selects a mesh-independent microstructure once it is resolved, we repeat the
computation with demonstration coefficients $q_c=1.0\times10^{-4}\,\mu$ and
$c_2=2.0\times10^{-4}\,\mu$ - prescribed directly rather than through material
values of $k$ and $\ell$ - at the slip orientation $\varphi=-1.2$
(misorientation $\theta=42.5^{\circ}$) and the mid-window shear $\gamma=0.39$,
$\gamma_B=0.778$. For these coefficients the closed-form layer profile of
Section~\ref{sec:GNB} has the width (full width at half maximum of
$|\nabla\beta\cdot\bs|$) $w_{\mathrm{1d}}=0.087$, resolved by $5.6$, $8.3$ and
$12.5$ elements on the $64^{2}$, $96^{2}$ and $144^{2}$ meshes.

Two regimes result. A resolvable wall --- the layer's full excess energy per
unit length, $2\mu\!\int p\,\mathrm{d}x_{2}+q_{c}|\sin\varphi|\,|\beta_{B}|$,
cf.\ Section~7 --- costs ${\approx}\,e/3.7$ at $\theta=19.6^{\circ}$ but only
${\approx}\,e/13$ at $\theta=42.5^{\circ}$. At the small misorientation this
is the size effect of the theory \citep{koster2015aformation} seen from
within: in a specimen only an order of magnitude larger than the internal
length, a handful of walls consumes the entire driving force, and the crystal
remains a single grain (confirmed by resolved runs at $\varphi=-1.4$,
$\gamma=\gamma_B/2$: the unseeded descent converges to the wall-free state
on $48^{2}$ and $96^{2}$ meshes, and a seeded one-pair laminate relaxes
back to it, the converged energies at fixed mesh agreeing to seven
significant figures).  At $\theta=42.5^{\circ}$ and $\gamma=0.39$, by contrast, the driving force
$e(\gamma)=9.9\times10^{-3}\,\mu$ exceeds the wall cost thirteenfold, and
the laminate forms \emph{spontaneously}: the descent from the affine state
bifurcates into a three-pair laminate without seeding or perturbation, and
seeded branches with one to four pairs all relax to the same six-interface
state, their converged energies agreeing to within one percent and each
lying below the bulk value $e(\gamma)$ of homogeneous slip --- hence,
\emph{a fortiori}, below the energy of any admissible state with uniform
slip in the bulk and its attendant boundary layers.

Table~\ref{tab:resolved} collects the mesh study. The interior wall count is
six on every mesh from $48^{2}$ to $144^{2}$; the measured wall width
descends monotonically onto the closed-form value $w_{\mathrm{1d}}=0.087$;
and the energy decrement per refinement falls from $3.4\%$ to $1.3\%$ -
percent-level residuals of the boundary-layer resolution, against the
$20$-to-$37\%$ of the unresolved computation (Table~\ref{tab:res-mesh}). The
count and the width are thus set by the internal length, not by the mesh:
the regularized functional attains its minimizer (Proposition~\ref{prop:existence}) and the computation converges to it. Figure~\ref{fig:resolved} shows the converged plastic slip and the wall profile at $144^{2}$. The walls are slightly inclined to the shear direction and the slip in the $B$ lamellae overshoots the well value ($\beta$ reaches $-0.93$ against $\beta_B=-0.78$), since the local shear in the constrained
bands exceeds $\gamma_B$; accordingly the measured peak of
$|\nabla\beta\cdot\bs|$ exceeds the one-dimensional prediction, which
idealizes the wall as a plane rank-one transition between the unperturbed
wells, while the width - the quantity selected by the regularization - matches
it. Finally, repeating the $144^{2}$ computations with the transverse
regularization $\tfrac{1}{2}\varepsilon_m c_2(\nabla\beta\cdot\bmm)^{2}$,
$\varepsilon_m=10^{-2}$, changes the total energy by $0.18\%$ - the value of
the added penalty itself - and the wall count, width and pattern not at all;
together with the boundedness of $\max|\nabla\beta\cdot\bmm|$ under
refinement, this rules out relaxation by transverse laminates, the numerical
counterpart of the missing transverse compactness discussed in
Remark~\ref{rem:slipbound}.

\begin{table}[htb]
\centering
\begin{tabular}{ccccc}
\hline
mesh & $E$ (unseeded) & interior walls & FWHM & decrement \\
\hline
$48^{2}$  & $9.208\times10^{-3}$ & 6 & $0.096$ & - \\
$64^{2}$  & $8.899\times10^{-3}$ & 6 & $0.091$ & $3.4\%$ \\
$96^{2}$  & $8.657\times10^{-3}$ & 6 & $0.089$ & $2.7\%$ \\
$144^{2}$ & $8.540\times10^{-3}$ & 6 & $0.084$--$0.088$ & $1.3\%$ \\
\hline
\end{tabular}
\caption{Resolved internal length ($q_c=10^{-4}\mu$, $c_2=2\times10^{-4}\mu$,
$\varphi=-1.2$, $\gamma=0.39$): total energy of the unseeded run, interior
wall count of the converged laminate, measured wall width against the
closed-form $w_{\mathrm{1d}}=0.087$, and energy decrement per refinement.
The wall count is mesh-independent and the energy decrement is at the percent
level, against $20$-to-$37\%$ in the unresolved regime of
Table~\ref{tab:res-mesh}. (On the finest mesh a vertical cut may register
five interior walls, the walls being slightly inclined; the seeded branches
give six on every mesh.)}
\label{tab:resolved}
\end{table}

\begin{figure}[htb]
\centering
\includegraphics[width=0.48\textwidth]{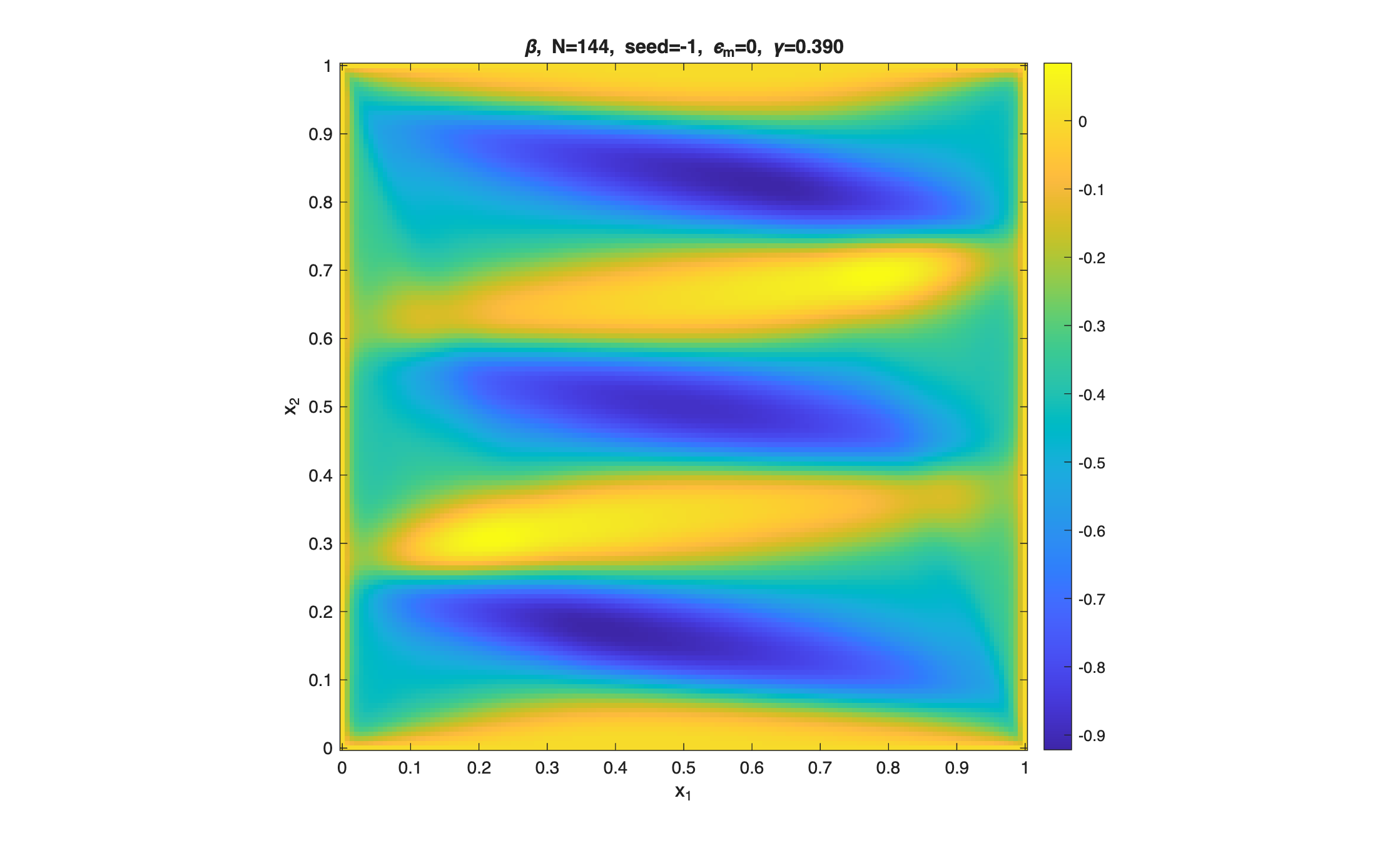}\hfill
\includegraphics[width=0.48\textwidth]{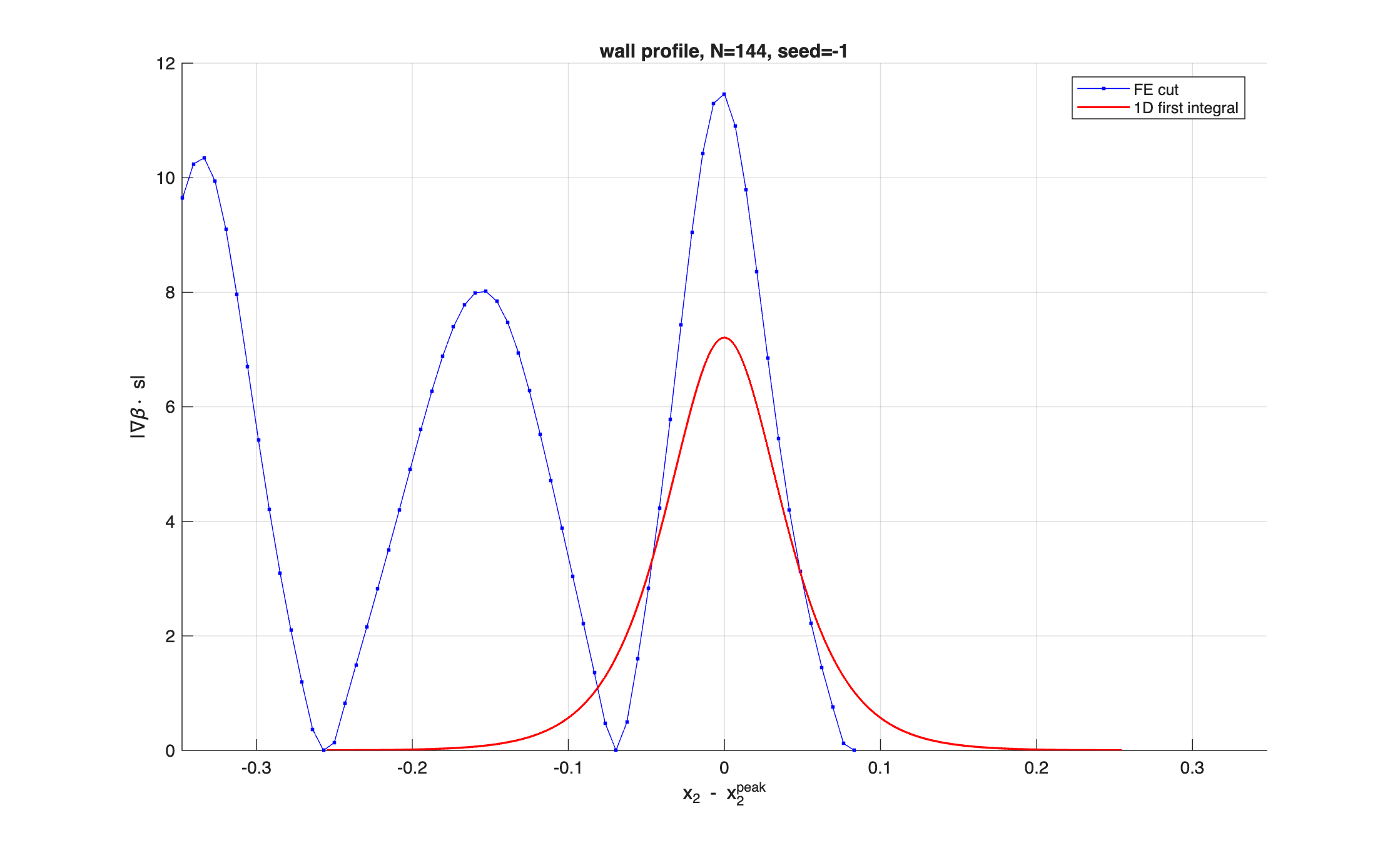}
\caption{Resolved internal length at $144^{2}$ ($\varphi=-1.2$,
$\gamma=0.39$). Left: converged plastic slip $\beta$ of the unseeded run -
the laminate of three $B$ lamellae forms spontaneously from the affine
state; the walls are slightly inclined to the shear direction, and the slip
in the $B$ lamellae overshoots the well value $\beta_B=-0.78$. Right: wall
profile $|\nabla\beta\cdot\bs|$ on a vertical cut against the closed-form
layer profile of Section~\ref{sec:GNB}: the measured width matches
$w_{\mathrm{1d}}=0.087$, while the peak exceeds the one-dimensional value in
proportion to the enlarged well-to-well jump of the constrained
two-dimensional state.}
\label{fig:resolved}
\end{figure}

\paragraph{Boundary energy across misorientation}
Running the scheme across slip orientations $\varphi$ (misorientation
$\theta=2\varphi+\pi$) yields, for each converged laminate, the dislocation
content $N/L$ (walls per unit width, read from the finite element solution).
Together with the sharp-wall boundary energy $\gamma_G=q_c|\beta_B|$ per unit
boundary length - a closed-form quantity of the model - these are collected in
Table~\ref{tab:misor} and Figure~\ref{fig:misor}. Unlike the total energy, the
closed-form $\gamma_G$ is mesh-independent, while $N/L$ - read from the $48^2$
laminate - reflects a single mesh; both nonetheless grow monotonically with
$\theta$ and give the physical characterization of the geometrically necessary
boundaries, in the spirit of the earlier boundary-energy
computations. The corresponding boundary thickness $h(\theta)$ is
obtained in Section~\ref{sec:GNB} from the closed-form layer potential
(Table~\ref{tab:hthickness}).

\begin{table}[h]\centering\small
\begin{tabular}{@{}cccc@{}}
\toprule
$\theta$ (deg) & $\varphi$ (rad) & $\gamma_G$ (units of $\mu L$) & $N/L$ \\
\midrule
 8.1 & $-1.50$ & $2.16\times10^{-8}$ & 2.27 \\
19.6 & $-1.40$ & $5.25\times10^{-8}$ & 2.94 \\
31.0 & $-1.30$ & $8.44\times10^{-8}$ & 6.42 \\
42.5 & $-1.20$ & $1.18\times10^{-7}$ & 7.31 \\
53.9 & $-1.10$ & $1.55\times10^{-7}$ & 7.76 \\
65.4 & $-1.00$ & $1.95\times10^{-7}$ & 8.01 \\
\bottomrule
\end{tabular}
\caption{Boundary energy $\gamma_G=q_c|\beta_B|$ (dislocation-wall energy per unit boundary length) and dislocation content $N/L$ (walls per unit width, read from
the converged $48^2$ laminate) versus misorientation angle
$\theta=2\varphi+\pi$. Both increase monotonically with $\theta$; the growth of
$\gamma_G$ mirrors the boundary-energy curve of
\citet{koster2015bformation}. The closed-form $\gamma_G$ is mesh-independent,
while $N/L$ is a single-mesh count.}
\label{tab:misor}
\end{table}

\begin{figure}[h]\centering
\includegraphics[width=\linewidth]{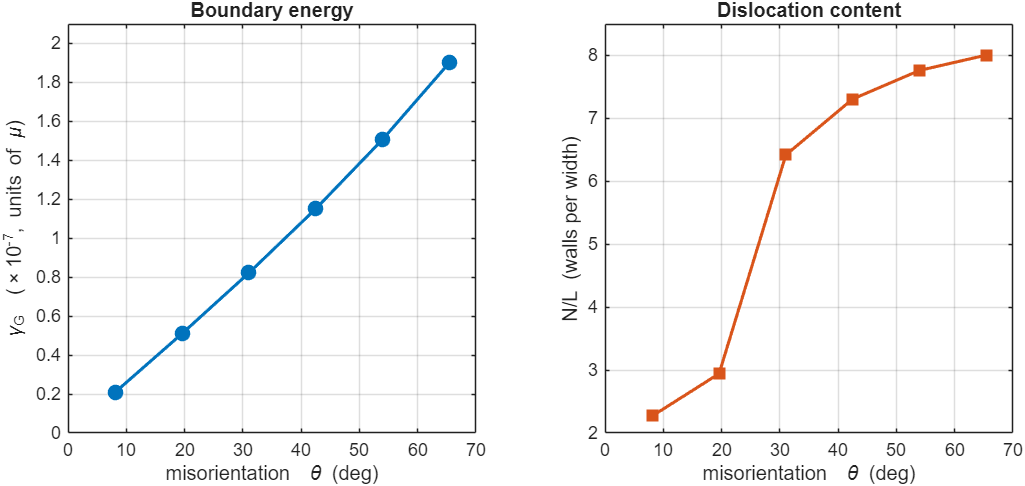}
\caption{Boundary energy $\gamma_G$ (left) and dislocation content $N/L$ (right)
versus misorientation angle $\theta$: both grow with $\theta$; $\gamma_G$ is
mesh-independent while $N/L$ is a single-mesh count.}
\label{fig:misor}
\end{figure}

\paragraph{Physical scales} The computation is dimensionless; the internal
length restores the physical scale. For aluminum ($\mu=26$~GPa) with Burgers
vector $b=0.25$~nm and saturated dislocation density
$\rho_s=10^{16}\,\mathrm{m}^{-2}$ \citep{koster2015bformation}, the internal length is $\ell=1/(b\rho_s)\approx400$~nm, and the value $\eta=\ell/L=1.9\times10^{-3}$ used here corresponds to a specimen $L\approx0.2$~mm. The computed walls have width of order $\sqrt{k}\,\ell$ - a few nanometers for the value of $k$ used in the runs - while the profile-based thickness for the material parameters
above is the $h=48$-to-$8$~nm of
Table~\ref{tab:hthickness}, close to their analytical $50$-to-$6$~nm; the grain
(laminate) spacing $L/N$ is of order tens of micrometers. Each geometrically
necessary boundary then contains
$N_d=|\beta_B||\sin\varphi|\,L/b\approx3\times10^{5}$ dislocations. In the same units the computed boundary energies of
Table~\ref{tab:misor} convert, through the factor $\mu L$ and at the dislocation
modulus of the runs, to $\gamma_G\approx0.1$-$1.1$~N/m across $\theta=8$-$65^\circ$ - the same range as the approximately $1$~N/m obtained analytically for
misorientations up to $50^\circ$ in the sharp-wall theory. The grain count follows the size effect $\varepsilon\sim\sqrt{\gamma_G L/\mu}$, i.e.\
$N\sim L/\varepsilon\sim\sqrt{\mu L/\gamma_G}$: because the boundary energy
$\gamma_G=q_c|\beta_B|\propto k$ grows with the dislocation modulus $k$, a larger $k$ produces thicker walls, coarser grains and hence fewer of them. A directly
mesh-converged grain count is precluded by the non-attainment (the walls lie below
the mesh size), so we report this scaling rather than a single-mesh integer; the
fixed-mesh trend with $k$ is nonetheless examined below.

\paragraph{Grain count versus the dislocation modulus}
Sweeping the dislocation modulus over $k=2\times10^{-4}$ to $8\times10^{-3}$ at
fixed cell, mesh and protocol ($64^{2}$, $\gamma\to0.216$) probes how the grain
count responds to the price of a wall (Figure~\ref{fig:ngraink}). The wall count
stays on a plateau $N\approx6$ for $k\le10^{-3}$, then coarsens: $N=5.9$, $5.4$
and $3.4$ at $k=2$, $4$ and $8\times10^{-3}$. The mechanism is visible in the
energy ledger: the wall energy grows with $k$ - linearly while the count is
constant, from $7.9\times10^{-7}$ at $k=2\times10^{-4}$ to $1.7\times10^{-5}$ at
$8\times10^{-3}$ - yet stays below $4\%$ of the elastic energy even at the
largest $k$, and more than two orders of magnitude below it on the plateau, so
the laminate is at first insensitive to $k$; once the wall bill becomes
appreciable the system trades walls for elastic energy (the elastic part rises
from $4.486\times10^{-4}$ to $4.552\times10^{-4}$ across the sweep) and the
microstructure coarsens. The direction agrees with the size-effect scaling
$N\sim\sqrt{\mu L/\gamma_G}$ with $\gamma_G\propto k$; a quantitative exponent
would require resolving the wall interior, which lies below the mesh size here,
and is left to future work at resolved internal length across the full misorientation range.

\begin{figure}[h]\centering
\includegraphics[width=0.62\linewidth]{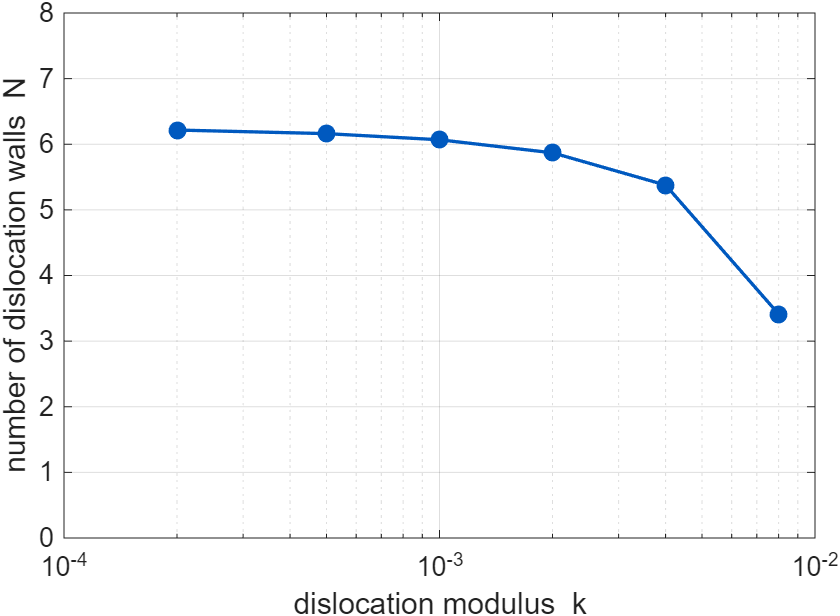}
\caption{Wall count $N$ versus the dislocation modulus $k$ at fixed cell, mesh
and load ($64^{2}$, $\gamma=0.216$): a plateau $N\approx6$ for $k\le10^{-3}$,
followed by coarsening once the wall energy becomes appreciable relative to the
elastic energy.}
\label{fig:ngraink}
\end{figure}

\begin{figure}[h]\centering
\includegraphics[width=\linewidth]{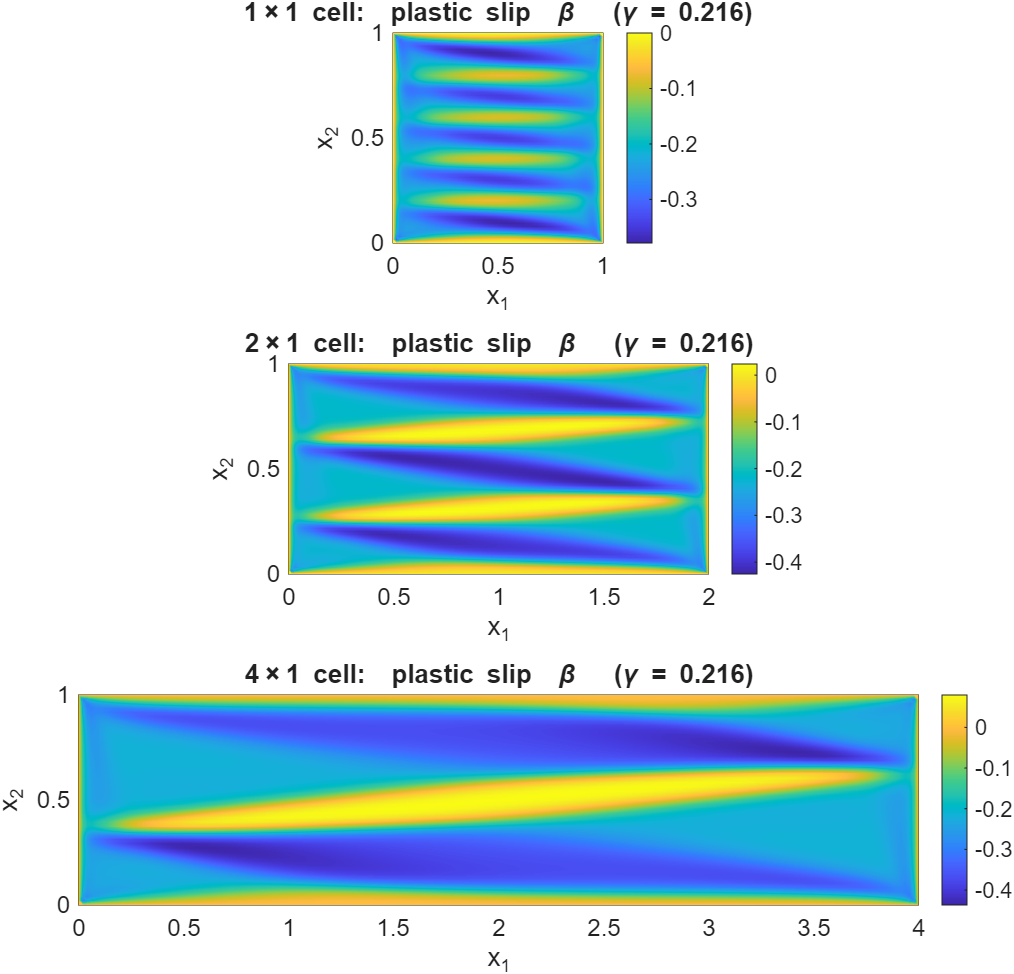}
\caption{Plastic slip $\beta$ at $\gamma=0.216$ on $1\times1$, $2\times1$ and
$4\times1$ cells (same $k$ and protocol): the laminate stays parallel to the
shear direction and coarsens with the cell length - about six, four and three
layers - following the size-effect scaling
$\varepsilon\sim\sqrt{\gamma_G L/\mu}$.}
\label{fig:rect}
\end{figure}

\paragraph{Elongated cells}
Repeating the computation on $2\times1$ and $4\times1$ cells (meshes
$128\times64$ and $192\times48$; the same $k=8\times10^{-5}$ and loading
protocol as the $1\times1$ baseline of Figure~\ref{fig:rect}, top) probes the
response to the specimen shape (Figure~\ref{fig:rect}). Three observations. First, the
layers remain parallel to the shear direction at every aspect ratio, with only a
slight inclination of the interior interfaces developing in the longest cell.
Second, the laminate \emph{coarsens} with the cell length: about six layers
across the height of the $1\times1$ cell, four on $2\times1$ and three on
$4\times1$ - in close agreement with the size effect
$\varepsilon\sim\sqrt{\gamma_G L/\mu}$ of \citet{koster2015bformation}, which predicts the layer count to fall by $1/\sqrt2\approx0.71$ and $1/2$ - matching the
observed reductions to within the integer count. (The $4\times1$ cell carries a
slightly coarser mesh, so its count is also mesh-influenced; the mesh-matched
$1\times1$ against $2\times1$ pair carries the clean comparison. On a fixed mesh
the count also depends mildly on the loading protocol - different protocols
select energetically near-degenerate laminates, consistent with
Remark~\ref{rem:notconvex} - which is why the mesh study above, run under a
different protocol, reports four layers on the same mesh; the three runs
compared here share one protocol.) Third, the mean energy density
falls with the aspect ratio, $4.49\times10^{-4}\to3.05\times10^{-4}\to
2.53\times10^{-4}$, and the share of the energy residing near the specimen
boundary falls from $46\%$ to $41\%$: the hard-device ends are expensive, and
lengthening the cell dilutes them, moving the mean energy toward the relaxed
laminate value.

\paragraph{Energy budget}
Decomposing the total energy of the working state ($1\times1$, $\gamma=0.216$; the six-lamella state of Figure 7, top)
by term and by region makes the relaxation picture quantitative. By term, the
elastic part carries $99.9\%$ of the total ($4.486\times10^{-4}$, against
$3.2\times10^{-7}$ for the bounded-variation wall term and $1.1\times10^{-9}$ for
the quadratic gradient term): the dislocation walls themselves are energetically
almost free, exactly as the smallness of $\gamma_G$ in Table~\ref{tab:misor}
anticipates. By region, a band of width $0.08$ along the specimen boundary -
$29\%$ of the area - holds $46\%$ of the energy, so the energy density near the
boundary is about twice that of the interior. The interior laminate is thus
nearly stress-free while the hard-device boundary layers concentrate the cost:
the numerical face of the vanishing relaxed energy $e^{\ast\ast}=0$ of
Section~\ref{sec:relax}, with the residual energy pushed into the constrained
layers.

\section{Conclusions}\label{sec:conclusions}

We have studied the formation of grain boundaries in ductile single crystals under plane-strain simple shear within continuum dislocation theory, replacing the non-polyconvex Saint-Venant--Kirchhoff energy of earlier treatments \citep{koster2015aformation,koster2015bformation} by the polyconvex Ciarlet–Geymonat energy, following the dissertation of \citet{koster2018modeling}, which first replaced the SVK energy of this setting by a polyconvex (neo-Hookean) one. This choice renders the deformation subproblem well posed and yields existence of minimizers for the coupled deformation--slip problem (Proposition~\ref{prop:existence}), while reproducing the same isotropic Hooke law in the small-strain limit and a condensed energy that remains a double well, now of lower algebraic degree than its Saint-Venant--Kirchhoff counterpart. The non-quasiconvexity of this energy drives the spontaneous formation of a lamellar microstructure whose interfaces are grain boundaries; the dislocation-density gradient regularizes these boundaries to a finite thickness and energy, closely matching the earlier analysis. A block-coordinate finite element scheme --- alternating a convex ADMM solve for the plastic slip with a Levenberg-regularized Newton solve for the deformation, together with an energy-based test for the onset of microstructure --- reproduces this lamellar structure numerically from a homogeneous initial state, without the one-dimensional ansatz of the analytical route, with wall count and width becoming mesh-independent once the internal length is resolved. The robustness of the results under the change of elastic energy indicates that grain-boundary formation is an intrinsic consequence of the geometric incompatibility between the imposed shear and the slip system; the scheme extends naturally to multiple slip systems, other loadings, and general two-dimensional geometries where closed-form solutions are unavailable. In particular, the interaction of multiple dislocation systems \citep{le2008plane,le2009plane} couples the slips and should resolve the transverse degeneracy noted in Section 8, rendering the multi-slip extension a more realistic test of the present picture. Natural next steps include a fully resolved boundary-layer profile at the material internal length, rate effects and dissipation, and the extension to three dimensions.
\bibliographystyle{unsrtnat}
\bibliography{refs}
\end{document}